\documentclass[%
12pt,
prereprint,
showpacs,
showkeys,
preprintnumbers,
amsmath,amssymb,
aps,
pra,
longbibliography,
notitlepage
]{revtex4-1}

\usepackage[dvipsnames]{xcolor}

\usepackage{mathptmx}

\usepackage{amssymb,amsthm,amsmath,bm}

\usepackage{float}
\makeatletter
\let\newfloat\newfloat@ltx
\makeatother

\usepackage[english]{babel}
\usepackage[utf8]{inputenc}

\usepackage{algorithm}
\usepackage{algpseudocode}

\usepackage{etoolbox}

\makeatletter
\newcommand*{\algrule}[1][\algorithmicindent]{%
	\makebox[#1][l]{%
		\hspace*{.2em}
		\vrule height .75\baselineskip depth .25\baselineskip
	}
}

\newcount\ALG@printindent@tempcnta
\def\ALG@printindent{%
	\ifnum \theALG@nested>0
	\ifx\ALG@text\ALG@x@notext
	\else
	\unskip
	\ALG@printindent@tempcnta=1
	\loop
	\algrule[\csname ALG@ind@\the\ALG@printindent@tempcnta\endcsname]%
	\advance \ALG@printindent@tempcnta 1
	\ifnum \ALG@printindent@tempcnta<\numexpr\theALG@nested+1\relax
	\repeat
	\fi
	\fi
}
\patchcmd{\ALG@doentity}{\noindent\hskip\ALG@tlm}{\ALG@printindent}{}{\errmessage{failed to patch}}
\patchcmd{\ALG@doentity}{\item[]\nointerlineskip}{}{}{} 
\makeatother

\usepackage{tikz}

\usetikzlibrary{calc,decorations.pathreplacing,decorations.markings,positioning,shapes,snakes,external}

\usepackage[breaklinks=true,colorlinks=true,anchorcolor=blue,citecolor=blue,filecolor=blue,menucolor=blue,urlcolor=blue,linkcolor=blue]{hyperref}
\usepackage{graphicx}
\usepackage{url}

\usepackage{soul}

\definecolor{TITLE}{rgb}{0,0,0}
\definecolor{midblue}{rgb}{0.00,0.0,0.80}
\definecolor{darkblue}{rgb}{0.00,0.00,0.45}
\definecolor{SECTION}{rgb}{0.50,0.00,1.00}
\definecolor{THM}{rgb}{0.8,0,0.1}
\definecolor{SEC}{rgb}{0,0,1}

\newtheorem{theorem}{{\color{THM} Theorem}}

\newtheorem{lemma}[theorem]{{\color{THM}Lemma}}

\newtheorem{corollary}[theorem]{{\color{THM}Corollary}}
\newtheorem{conjecture}[theorem]{{\color{THM}Conjecture}}
\theoremstyle{definition}
\newtheorem{definition}[theorem]{{\color{THM}Definition\ }}
\newtheorem{criterion}[theorem]{{\color{THM}Criterion\ }}

\begin{document}
	
	\title{Noncontextual coloring of orthogonality hypergraphs}

	\author{Mohammad H. Shekarriz}
	\email{mhshekarriz@yazd.ac.ir}
	\affiliation{Department of Mathematics, Yazd University, 89195-741, Yazd, Iran}
	
	\author{Karl Svozil}
	\email{svozil@tuwien.ac.at}
	\homepage{http://tph.tuwien.ac.at/~svozil}
	
	\affiliation{Institute for Theoretical Physics,
		Vienna  University of Technology,
		Wiedner Hauptstrasse 8-10/136,
		1040 Vienna,  Austria}

	\date{\today}
	
	\begin{abstract}
		We discuss representations and colorings of orthogonality hypergraphs in terms of their two-valued states interpretable as classical truth assignments. Such hypergraphs, if they allow for a faithful orthogonal representation, have quantum mechanical realizations in terms of intertwined contexts or maximal observables that are widely discussed as empirically testable criteria for contextuality. Reconstruction is possible for the class of perfectly separable hypergraphs. Colorings can be constructed from a minimal set of two-valued states. Some examples from exempt categories are presented that either cannot be reconstructed by two-valued states or whose two-valued states cannot yield a chromatic number that is equal to the maximal clique number.
	\end{abstract}
	
	
	\maketitle

	\section{Reconstructability and coloring of quantum contexts}
	
	Current arguments validating quantum contextuality start with configurations of quantum observables whose intertwined contexts can be structurally expressed as (hyper)graphs. These structures, when interpreted classically, exhibit features and predictions that contradict the outcome of the respective quantum observables. Therefore, (hyper)graph techniques proved to be a useful tool for quantum-vs-classical modeling.
	In what follows we shall pursue a related question: when is it still possible to characterize those (hyper)graphs by purely classical means? That is, given all two-valued states interpretable as classical truth assignments, is it possible to reconstruct the (hyper)graphs that, up to isomorphisms, characterize those observables?
	Inspired by and extending Theorem~0 of the 1967 paper of Kochen and Specker~\cite{kochen1} we shall present demarcation criteria for (hyper)graph reconstructability.
	
	We shall also explore the connection of two-valued states and coloring of (hyper)graphs, thereby presenting criteria for (re)construction of colorings and the chromatic number of a (hyper)graph from the set of its two-valued states. Thereby two-valued states may yield a systematic, constructive way to both colorings and/neither the reconstruction of (hyper)graphs.
	
	
	The chromatic properties of such (hyper)graphs are directly related to their (non)classical aspects.
	For instance, in quantum logic~\cite{birkhoff-36}
	certain colorings---with the number of colors equal to the clique number that can be identified with the Hilbert space dimension---can be ``reduced'' to or ``collapsed'' into two-valued dispersionless states, which in turn are interpretable as noncontextual classical truth assignments of the elementary propositions represented by vertices of such graphs. If the chromatic number exceeds  the Hilbert space dimension then no uniform, global two-valued state and also no corresponding uniform, global classical truth value assignment exist.
	Thereby, the chromatic number signifies important properties and (non)existence of classical interpretations of the aforementioned
	configurations of observables; in particular, arrangements of observables realizable by quantum means.
	
	One example of the usefulness of these reduction techniques is the construction of a dense yet discontinuous coloring
	of the rational sphere~\cite{meyer:99}---with the resulting states and propositions
	represented by unit vectors with rational coordinates---based on Pythagorean triples~\cite{havlicek-2000}.
	Thereby, a two-valued dispersionless state can be straightforwardly obtained by identifying all but one colors
	of colorings of the sphere~\cite{godsil-zaks}.
	
	The usefulness of colorings is not restricted to reductions of the colors but can be extended to certain properties of operators.
	The chromaticity of observables within a given context---which in quantum mechanics can be essentially identified with an orthonormal basis
	of the associated finite-dimensional Hilbert space---can be associated
	with particular types of spectral forms of maximal operators~\cite[\S~84]{halmos-vs} formed by the
	elements of the contexts.
	Thereby, unit vectors are interpreted as the orthogonal projection operators formed by the respective dyadic products.
	Those orthogonal projection operators within any given context are mutually orthogonal and can be
	inserted into the spectral sum of a (normal, or, more specifically) self-adjoint operator.
	The respective (real) eigenvalues can be identified with some real-valued encoding of the color or chromaticity of the element of the context.
	Thereby a uniform way of defining (intertwining) context identifiable with maximal quantum operators representing quantum observables is obtained.
	
	If the chromatic number equals the dimension of the Hilbert
	space---which is, at the same time, identical to the clique number of the graph---this yields
	a uniform way of defining (maximal) observables even among complementary observables and physical properties
	(associated with nonidentical contexts).
	However, if the chromatic number exceeds the dimension of Hilbert space,
	the construction yields entirely new potential features of observables:
	if one insists on the uniform global simultaneous (yet counterfactual~\cite{specker-60}) existence of observables
	within such structures, then consistency dictates the abandonment of uniform eigenvalues associated with observables within such contexts.
	This holds even though the chromatic number of a (hyper)graph and its respective encoding by real values for the spectral sum within a single such context
	cannot exceed the dimension of the pertinent Hilbert space.
	As of today, neither such ``omni-realistic'' escape nor generalization of the Kochen-Specker theorem has been discussed or observed---so we might assume that it is an inapplicable option. In addition, yet it is feasible in principle.
	
	In what follows, we shall first develop the necessary nomenclature and then present some criteria and results with regard to the reconstruction of (hyper)graphs representing
	logics---aka, collections of (intertwined) contexts containing a uniform number of elementary, atomistic propositions.
	We shall also find criteria for the algorithmic (constructive) generation of colorings
	by the set of two-valued states on such (hyper)graphs or logics.

	\section{Nomenclature}

	\subsection{Hypergraphs}
	
	The following terms are used by authors synonymously:
	context, block, (maximal Boolean) subalgebra, (maximal) clique, complete subgraph and \emph{hyperedge}. The same is applied to the terms atom, element and \emph{vertex}.
	
	Greechie has suggested~\cite{greechie:71} to (amendments are indicated by square brackets ``[$\ldots$]'')
	\begin{quote}
		[$\ldots$]
		present [$\ldots$] lattices as unions of [contexts]
		intertwined or pasted together in some fashion
		[$\ldots$]
		by replacing, for example, the $2^n$ elements in the Hasse diagram of the power set
		of an $n$-element set with the [context aka] complete [sub]graph [$K_n$] on $n$ elements.
		The reduction in numbers of elements is considerable but the number of remaining ``links''
		or ``lines'' is still too cumbersome for our purposes.
		We replace the [context aka] complete [sub]graph on $n$ elements by a single smooth curve (usually a straight line)
		containing $n$ distinguished points. Thus we replace $n(n + 1)/2$ ``links'' with a single smooth curve.
		This representation is propitious and uncomplicated provided that
		the intersection of any pair of blocks contains at most one atom.
	\end{quote}
	
	In what follows, we shall refer to such a structure as {Greechie diagram}~\cite{kalmbach-83} (References~~\cite{Greechie1968,svozil-tkadlec,Mckay2000,Pavicic-2005,Bretto-MR3077516,2018-minimalYIYS}
	contain variants thereof). The Greechie diagrammatical representation of such, possibly intertwined, collection of blocks, is a {\em hypergraph}, which is a well-known structure in discrete mathematics. We shall briefly introduce the required terms here, but an interested reader might take a look at Ref.~\cite{Bretto-MR3077516} for further theory of hypergraphs.
	
	A hypergraph  $H$  is an ordered pair $H=(V ;E)$ where $V=V(H)$ is the set of vertices and $E=E(H)$ is a family of subsets of $V$ called the hyperedges.
	It depicts a collection of quantum contexts faithfully represented in an $n$-dimensional Hilbert space, whereby the vertices are identified with elementary quantum propositions whose label assignments are in terms of vectors (or with their respective one-dimensional orthogonal projection operators), and the hyperedges are identified with quantum contexts. An $n$-subset of atoms forms a hyperedge if its elements are mutually orthogonal.
	
	Let $H = (V;E)$ be a hypergraph. The \emph{2-section}  of $H$ is a graph, denoted by $[H]_2$, whose vertices are the same as $V(H)$, and two distinct vertices form an edge if and only if they are on the same hyperedge of $H$. One may think of the 2-section of a hypergraph as the graph associated with the hypergraph. A hypergraph $H$ is called \emph{conformal} if any maximal clique (with respect to the inclusion) of its 2-section $[H]_2$ is on a hyperedge of $H$. If all the hyperedges of a hypergraph $H$ consist of exactly $n$ elements, then $H$ is called \emph{$n$-uniform} \cite{Bretto-MR3077516}.

	Quantum orthogonality hypergraphs are conformal $n$-uniform. Often these hypergraphs are referred to and depicted by their associated graphs; that is, by their 2-sections. Quantum contexts are represented by hyperedges of these quantum hypergraphs, that is, by the maximal cliques of their 2-sections. We also reserve the letter $n$ for the clique number of those 2-sections, so we always mean $n=\omega ([H]_2 )$, which is a constant integer.
	
	To represent orthogonality hypergraphs, we shall concentrate on Greechie diagrams which are pasting~\cite{Greechie1968} constructions~\cite[Chapter~2]{greechie-66-PhD}
	of a homogeneous  single type
	of contexts $K_n$
	where the clique number $n$ is fixed. This means that every hyperedge is shown by a straight line segment or, more general, by a smooth curve which has exactly $n$ elements as vertices on it; that is, the hypergraphs are conformal $n$-uniform.
	
	\subsection{Vertex labeling by vectors}
	
	A \emph{vertex labeling} of an orthogonality (hyper)graph $H$ is a function $f:V(H)\longrightarrow L$ that assigns labels from a set $L$ to vertices of $H$.
	
	\begin{criterion}\label{orthogonality}
		There is a vertex labeling $\bm{x}:V(H)\longrightarrow L$ that assigns a set of $k$ mutually non-colinear vectors in an $n$-dimensional Hilbert space $L$ to vertices of $H$ such that any pair of vertices $a$ and $b$ are adjacent if and only if $\bm{x}(a)$ is orthogonal to $\bm{x}(b)$.
	\end{criterion}
	
	Such a vertex labeling is called an \emph{$n$-dimensional faithful orthogonal representation}~\cite{lovasz-79,lovasz-89,Portillo-2015} (FOR) for~$H$. Any vertex labeling corresponds to a quantum realization in terms of the elementary propositions corresponding to its vector labels:
	every unit vector $\bm{x}$ spans a one-dimensional subspace of Hilbert space that is the orthogonal projection onto that subspace of the Hilbert space.
	
	\subsection{Coloring}
	
	A hypergraph coloring of $H$ is a \emph{proper vertex coloring} which associates colors to vertices of $H$ so that every two vertices lying on a hyperedge receive different colors. That is, the $n$ distinguished points of any single smooth curve in the hypergraph have $n$ different colors.
	The coloring is noncontextual; that is, the coloring of atomic elements common to two or more contexts (intertwining there)
	is independent of the context.
	
	For a hypergraph $H$, the \emph{chromatic number} $m=\chi(H)$ is the minimum number of colors required for a proper coloring of vertices of $H$. Obviously the clique number $n=\omega([H]_2 )$ is a lower~bound. If these numbers are the same, that is, if $m=\chi(H)=\omega([H]_2 )=n$, then one could obtain two-valued measures from colorings by ``projecting'' one of the colors into the value 1, and all the other $n-1$ colors into the value 0~\cite{godsil-zaks,meyer:99,havlicek-2000}. A hypergraph $H$, whose chromatic and clique numbers are equal, is called here \emph{semi-perfect}.

	Finite examples for which the chromatic number exceeds the clique number, that is, $m>n$,
	are the logical structures involved in proofs of the Kochen-Specker theorem.
	Explicit constructions are, for instance,
	$\Gamma_2$ of Ref.~\cite{kochen1},
	and the configurations enumerated in
	Figure~9 of~\cite{svozil-tkadlec},
	Figure~1--3 of~\cite{tkadlec-00},
	Ref.~\cite{cabello-96},
	and Table~I and Figure~2 of Ref.~\cite{2015-AnalyticKS},
	among numerous others
	which have a faithful orthogonal representation~\cite{lovasz-79,lovasz-89,Portillo-2015}
	in ``small dimensions'' greater than two.
	
	\subsection{Two-valued states}
	
	A \emph{state} $t$ on a conformal $n$-uniform orthogonality hypergraph $H$ is a mapping $t:\, V(H) \to [0,1]$ such that for any hyperedge $h$, we have $\sum_{v\in h} t(v)=1$. A \emph{two-valued state} is a state with values in $\{0,1\}$.

	A conformal $n$-uniform orthogonality hypergraph has a \emph{separable set of two-valued states} if for any distinct pair of vertices $u$ and $v$, there is at least one two-valued state, say $t$, such that $t( u ) \neq t(v)$ \cite{svozil-tkadlec}. A hypergraph $H$ has a \emph{unital} set of two-valued states if for each vertex $v \in V(H)$ there is a state $t$ for which we have $t(v)=1$. A hypergraph $H$ is said to be \emph{separable} (respectively, unital) if it has a separable (respectively, unital)
set of two-valued states on its vertices.
	
	A ``true implies false set'' $(a,b)$-TIFS (gadget~\cite{tutte_1954,SZABO2009436,Ramanathan-18}) is a conformal $n$-uniform orthogonality hypergraph $H$ containing two vertices $a$ and $b$  such that for all two-valued states of $H$, we have that $b$ is true only if $a$ is false.
	We call $a$ the ``head'' and $b$ the ``tail'' of $H$.

	Similarly, a ``true implies true set'' $(c,d)$-TITS (gadget) is a conformal $n$-uniform orthogonality hypergraph $H'$ containing two vertices $c$ and $d$ such that for all two-valued states of $H'$, we have that $d$ is true whenever $c$ is true~\cite{2018-minimalYIYS}
	(in this case the converse need not be true as both $c$ and $d$ could be false, or $d$ could be true and $c$ false). That is, $c$ true implies $d$ true.
	
	\subsection{Completion criterion}
	
	The hypergraphs considered here are assumed to be on $k$ vertices, and all their hyperedges (contexts) uniformly contain exactly $n$ vertices. Equivalently, we can assume that our objects are connected graphs such as $G$ on $k$ vertices with clique number $\omega (G)=n$, with the assumption that every vertex $v\in V(G)$ lies on at least one maximal clique (context) of size $n$. This assumption is not always necessary, but it is not harmful to our argument, as we can always add vertices to those contexts that have less number of vertices. Therefore, we can state \emph{the completion criterion} as follows:
	
	\begin{criterion}\label{completion}
		If $a$ and $b$ are adjacent in an orthogonality hypergraph $H$, then there is a hyperedge that contains $a$ and $b$ along with $n-2$ other vertices.
	\end{criterion}
	
	In other words, the completion criterion says that there must be no hyperedge with less than $n$ vertices. The importance of this simple criterion becomes clear in Section \ref{Rec-PSH}. Note that for any two-valued state of the hypergraph $H$, exactly one of these $n$ vertices is assigned true.
	
	\subsection{Separability, set representability and reconstructability}
	
	Another important criterion is about classical structure-preserving representability of quantum observables---that is, the classical representability of quantum logics---in terms of two-valued states. According to Theorem~0 of Kochen and Specker~\cite[p.~67]{kochen1}
	the possibility to ``separate'' and make a distinction between two arbitrary vertices of a (hyper)graph by at least one of its two-valued states is equivalent to homomorphic---that is, structure-preserving---embedability of the quantum observables into a ``larger'' classical Boolean algebra.
	
	An early example of a nonseparable (hyper)graph that is not reconstructable from its set of two-valued states is  $\Gamma_3$ introduced by Kochen and Specker~\cite[p.~70]{kochen1}.
	There is another hypergraph, depicted in Figure~5 of Ref.~\cite{svozil-2017-b} sharing the same Travis~\cite{travis-mt-62} matrix~\cite{greechie-66-PhD}---i. e., the same set of two-valued states---which is not isomorphic to the hypergraph corresponding to $\Gamma_3$.
	
	However, separability does not imply graph theoretic reconstructibility. One reason for this is that end points of TIFS gadgets might not be separable from orthogonal vertices. We shall explicitly present an example (depicted in Figure~\ref{TIFS-non-Rec}) of such a configuration, although we are not able to present a hypergraph that has a faithful orthogonal representation in terms of vertex labeling by vectors.
	
	In general a (hyper)graph is \emph{reconstructable} if the Travis matrix---that is, the set of two-valued states---determines or encodes it completely.
	More explicitly, a reconstructable (hyper)graph is, up to isomorphism,  determined by its Travis matrix---that is, effectively, by permutation of its vertices (the column vectors of the Travis matrix) or the ordering of two-valued states (the row vectors of the Travis matrix). Equivalently, a hypergraph $H$ is reconstructable from its Travis matrix $T_H$ if, for any hypergraph $G$ whose Travis matrix $T_G$ is \emph{equivalent} to $T_H$ --- that is, $T_G$ can be obtained from $T_H$ by permutations of rows and columns --- the (hyper)graphs $H$ and $G$ are isomorphic.

	The only connected conformal 1-uniform orthogonality (hyper)graph is the trivial singleton whose two-valued states consists of only one state that assigns true to the only vertex. In addition, the only connected conformal 2-uniform orthogonality hypergraph --- which is actually a graph --- is $K_2$ whose Travis matrix is equivalent to the $2\times 2$ identity matrix. This is because all other connected graphs have a vertex, say $a$, adjacent to at least two other vertices, say $b$ and $c$, and therefore any representation of the graph into a 2-dimensional Hilbert space has to assign collinear vectors to $b$ and $c$. Hence, hereinafter, we suppose that the dimension of our space $n$, which is the uniform number of vertices on hyperedges, is greater than 2.

	\section{Structure reconstruction from two-valued states}\label{construction-conj}
	
	Suppose that we have a table of all two-valued states of a hypergraph $H$. Under what condition(s) can we reconstruct the graph theoretical structure of $H$  from the set of its two-valued states?  This is what we are going to examine next.
	
	\subsection{Perfectly separable hypergraphs}
	
	Let $H$ be a hypergraph with $V(H)=\{a_1, a_2,\ldots, a_k\}$, and its set of two-valued states contains $s$ elements [which can be shown by saying that $nTS(H)=s$]. Its Travis matrix $T(H) =[t_{ij}]_{s\times k}$ enumerates all two-valued states which are represented by row vectors on the vertices of $H$, arranged in the columns, such that each vertex corresponds to one column. Then, \emph{separability} of $H$ can be extended to the following statement:
	
	\begin{definition}\label{separability}
		The hypergraph $H$ is \emph{perfectly separable} if and only if for all pairs $a_i$ and $a_j$ of vertices of $H$ we obtain the following:
		\begin{itemize}
			\item[1.] There is $r_1 \in \{1,\ldots,s\}$ such that $t_{r_1 i}=0$ and $t_{r_1 j}=1$.
			\item[2.] There is $r_2 \in \{1,\ldots,s\}$  such that $t_{r_2 i}=1$ and $t_{r_2 j}=0$.
			\item[3.] If $a_i$ and $a_j$ are not adjacent in $H$, then there is an $r_3 \in \{1,\ldots,s\}$  such that $t_{r_3 i}=1$ and $t_{r_3 j}=1$.
		\end{itemize}
	\end{definition}
	
	It might seem that an additional condition like
	``If $a_i$ and $a_j$ are not adjacent in $H$, then there is an $r_4 \in \{1,\ldots,s\}$ such that $t_{r_4 i}=0$ and $t_{r_4 j}=0$''
	would make a difference and be independent of conditions 1--3. However, this additional condition can be deduced from condition 1 or 3 because $n\geq 3$ and so there is another vertex $a_k$ which is adjacent to $a_j$, and if $a_i$ is also adjacent to $a_k$, condition 1 implies it, and if $a_i$ and $a_k$ are not adjacent, condition 3 implies it. Thus, we did not include it in Definition~\ref{separability}.
	
	Note that a hypergraph is separable if item 1 or 2 of Definition \ref{separability} holds. Therefore, every perfectly separable orthogonality hypergraph is also separable, but the converse is not always true. Elementary counterexamples are true-implies-false gadgets such as the Specker bug depicted in Figure~\ref{2020-f-SpeckerBug}.
	
	\subsection{Reconstruction of perfectly separable hypergraphs}\label{Rec-PSH}
	
	We know that whenever $a_i$ and $a_j$ are adjacent in $H$ (or, equivalently, they are on the same context), then $t_{ri}=1$  implies $t_{rj}=0$ for $r=1,\ldots,s$. We want to know under what conditions the converse is true as well, i.e., the following \emph{adjacency criterion} holds:
	\begin{criterion}
		\label{AdCr} For all $r=1,\ldots,s$, if $t_{ri}=1$ implies $t_{rj}=0$
		then $a_i$ and $a_j$ are adjacent.
	\end{criterion}
	
	The next theorem asserts that when $H$ is a perfectly separable hypergraph, we can always reconstruct $H$ from the information presented in its Travis matrix $T(H)$ and the adjacency criterion. In addition, conversely, a reconstructable hypergraph using Criterion \ref{AdCr} has to be a perfectly separable hypergraph. In other words, Conjecture \ref{c2} is true for perfectly separable hypergraphs.
	
	\begin{theorem}\label{reconstruct}
		Let $H$ be a hypergraph on $\{a_1, a_2,\ldots, a_k\}$ whose Travis matrix $T(H)$ is available and $\omega([H]_2 )=n\geq 3$. Moreover, suppose that $H'$ is the hypergraph on $\{a_1, a_2,\ldots, a_k\}$ whose adjacency is defined by Criterion \ref{AdCr}. Then $H=H'$ if and only if $H$ is perfectly separable.
	\end{theorem}
	\begin{proof}
		First suppose that $H$ is perfectly separable. Then, because of item 3 of Definition \ref{separability}, for all $a_i$ and $a_j$ that are not adjacent in $H$ there is a row $r$ in $T(H)$ such that $t_{r i}=1$ and $t_{r j}=1$. Therefore, Criterion \ref{AdCr} in $H'$ can only be satisfied for those vertices that are already adjacent in $H$. Consequently, $H=H'$.
		
		Conversely, suppose that $H=H'$ and $a_i$ and $a_j$ are two distinct non-adjacent vertices. Then the adjacency criterion does not meet for $a_i$ and $a_j$. As a result, there is $r$, $1\leq r\leq s$ such that $t_{ri}=1$ and $t_{rj}=1$, i.e., item 3 of Definition \ref{separability} already holds for $a_i$ and $a_j$. Moreover, if item 1 (or 2) of Definition~\ref{separability} does not hold for $i$ and $j$, then the adjacency criterion makes $a_i$ adjacent to all vertices already adjacent to $a_j$, which cannot happen unless $a_i = a_j$, or $a_i$ is colinear with $a_j$, which is a contradiction.
		
		Therefore, for every pair of non-adjacent vertices of $H$, items 1 to 3 of Definition \ref{separability} hold. If for any pair of adjacent vertices in $H$, namely $a_i$ and $a_j$, both items~1 and~2 of Definition \ref{separability} are also true, then it can be inferred that $H$ must be perfectly separable.
		
		To see this, first notice that if one of items 1 or 2 does not hold for adjacent vertices $a_i$ and $a_j$, then one of them, say $a_j$, is always assigned $0$. Therefore, because $H=H'$, and $H'$ is constructed via Criterion \ref{AdCr}, $a_j$ must be adjacent to any vertex that is assigned $1$ at least once. Hence, we can distinguish the following two cases:
		
		\begin{itemize}
			\item[Case 1.] \emph{The vertex $a_j$ belongs to all contexts of $H$}. Then there is a two-valued state that assigns $1$ to $a_j$ and $0$ to all other vertices, a contradiction to our assumption that $a_j$ is always assigned $0$.
			
			\item[Case 2.] \emph{There is at least one context, $\mathcal{A}$, which does not contain $a_j$}. Since $H=H'$, only one vertex of $\mathcal{A}$, say $a_p$, can be assigned always $0$. Then $a_j$ is adjacent to at least $n-1$ vertices of $\mathcal{A}$ other than $a_p$, and consequently, $a_j$ and $a_p$ are colinear, a contradiction.
		\end{itemize}
		Therefore, items 1 and 2 of Definition \ref{separability} also hold for every pair of adjacent vertices of $H$. Now the conclusion is evident.
	\end{proof}
	
	From the proof of Theorem \ref{reconstruct} we see that, for a reconstructable hypergraph $H$, the most important factor is that at least item 1 or 2 is true for every pair of non-adjacent vertices. It should be noted that item 2 is actually the contraposition for item 1. Yet conversely, if  we only want to imply Criterion \ref{AdCr}, we see that for separable hypergraphs which are not perfectly separable, the reconstruction produces some extra hyperedges for $H'$ and makes it different from $H$. For example, when we try to reconstruct the Specker bug from its two-valued states using the adjacency criterion, we produce an extra hyperedge on two vertices, as depicted in Figure~\ref{Spec-Rec}.
	
	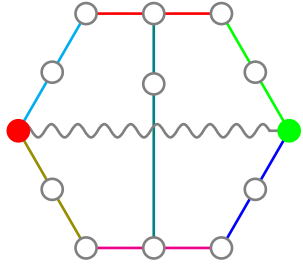
\begin{figure}
		\begin{center}
			\begin{tikzpicture}  [scale=0.6]
				
				\newdimen\ms
				\ms=0.07cm
				
				\tikzstyle{every path}=[line width=1pt]
				
				\tikzstyle{c3}=[circle,inner sep={\ms/8},minimum size=6*\ms]
				\tikzstyle{c2}=[circle,inner sep={\ms/8},minimum size=4*\ms]
				\tikzstyle{c1}=[circle,inner sep={\ms/8},minimum size=0.8*\ms]
				
				\newdimen\R
				\R=30mm     
				
				
				
				\path
				({ 180 - 0 * 360 /6}:\R      ) coordinate(1)
				({ 180 - 30 - 0 * 360 /6}:{\R * sqrt(3)/2}      ) coordinate(2)
				({ 180 - 1 * 360 /6}:\R   ) coordinate(3)
				({ 180 - 30 - 1 * 360 /6}:{\R * sqrt(3)/2}   ) coordinate(4)
				({ 180 - 2 * 360 /6}:\R  ) coordinate(5)
				({ 180 - 30 - 2 * 360 /6}:{\R * sqrt(3)/2}  ) coordinate(6)
				({ 180 - 3 * 360 /6}:\R  ) coordinate(7)
				({ 180 - 30 - 3 * 360 /6}:{\R * sqrt(3)/2}  ) coordinate(8)
				({ 180 - 4 * 360 /6}:\R     ) coordinate(9)
				({ 180 - 30 - 4 * 360 /6}:{\R * sqrt(3)/2}     ) coordinate(10)
				({ 180 - 5 * 360 /6}:\R     ) coordinate(11)
				({ 180 - 30 - 5 * 360 /6}:{\R * sqrt(3)/2}     ) coordinate(12)
				;
				
				
				\draw [color=cyan] (1) -- (2) -- (3);
				\draw [color=red] (3) -- (4) -- (5);
				\draw [color=green] (5) -- (6) -- (7);
				\draw [color=blue] (7) -- (8) -- (9);
				\draw [color=magenta] (9) -- (10) -- (11);    %
				\draw [color=olive] (11) -- (12) -- (1);    %
				\draw [color=teal] (4) -- (10)  coordinate[pos=0.3]  (13);
				\draw [color=gray, snake=coil,segment aspect=0] (1) -- (7);
				
				%
				%
				\draw (1) coordinate[c2,draw=red,fill=red];   %
				\draw (2) coordinate[c2,draw=gray,fill=white];    %
				\draw (3) coordinate[c2,draw=gray,fill=white]; %
				\draw (4) coordinate[c2,draw=gray,fill=white];  %
				\draw (5) coordinate[c2,draw=gray,fill=white];  %
				\draw (6) coordinate[c2,draw=gray,fill=white];
				\draw (7) coordinate[c2,draw=green,fill=green];  %
				\draw (8) coordinate[c2,draw=gray,fill=white];  %
				\draw (9) coordinate[c2,draw=gray,fill=white];
				\draw (10) coordinate[c2,draw=gray,fill=white];  %
				\draw (11) coordinate[c2,draw=gray,fill=white];  %
				\draw (12) coordinate[c2,draw=gray,fill=white];
				\draw (13) coordinate[c2,draw=gray,fill=white];  %
				%
			\end{tikzpicture}
		\end{center}
		\caption{\label{Spec-Rec}
			Reconstruction of the Specker bug using Criterion \ref{AdCr} which gives the gray snake-shaped extra hyperedge that contains only two (the red and green) vertices. This extra hyperedge can be easily eliminated if we use  Criterion \ref{completion}.
		}
	\end{figure}

	However, one might think that we may be able to reconstruct $H$ from $H'$ by finding the extra hyperedge(s) that must be eliminated. To do this, suppose that $\omega([H]_2 )=n\geq 3$ and every vertex is on a hyperedge of size $n$. Then the reconstructed $H'$ must have the same property, i.e., we must have Criterion \ref{completion} for all hyperedges of $H'$, and if a hyperedge does not meet this criterion, then it has to be eliminated. This ``unwanted'' adjacency between $a_i$ and $a_j$ in $H'$ appear only when these two are not on the same context but whenever one is assigned $1$ the other has to be $0$ and vice versa---a typical situation in a TIFS gadget---like what is happening when we try to reconstruct the Specker bug from its two-valued states. However, we show in Sec.~\ref{Rec-B(H)} that this method cannot always guarantee that we can identify these unwanted hyperedges from the table of two-valued states.

	\subsection{Examples of non-reconstructable hypergraphs using Criteria \ref{completion} and  \ref{AdCr}}\label{Rec-B(H)}
	
	Knowing that an orthogonality hypergraph is reconstructable from its  two-valued states gives us the opportunity to find perpendicular pairs of propositions without actually measuring angles between them.
	Recall that, for two-valued states, whenever a vertex is assigned true, all other vertices adjacent to it have to be assigned false. This adjacency criterion (Criterion~\ref{AdCr}) is the most evident property that every pair of adjacent vertices has.
	
	For a nonseparable hypergraph, however, this does not let us to reconstruct the hypergraph because for two distinct vertices $a$ and $b$ with $t(a)=t(b)$ for every two-valued state $t$, it cannot be understood whether or not $a$ is adjacent to all the neighbors of $b$.
	Therefore, it might be tempting to speculate that separability is a ``good'' criterion for reconstructability. Indeed, it might be tempting to claim the following general conjecture, against which we shall shortly present a counterexample.
	
	\begin{conjecture} \label{c2} A hypergraph $H$ is reconstructable from the table of its two-valued states if and only if $H$ is separable.
	\end{conjecture}
	
	From Theorem~\ref{reconstruct} we know that for perfectly separable hypergraphs Conjecture~\ref{c2} holds.
	Can we generalize this to separability?
	
	In what follows, we present a counterexample: a hypergraph that is separable but not perfectly separable. For the sake of a contradiction with Conjecture~\ref{c2} we introduce a configuration for which adjacency of vertices cannot be determined by the enumeration of all two-valued states (the Travis matrix) alone. The argument uses a maximal clique number $n=3$, but can be generalized to higher clique numbers.
	
	A counterexample to Conjecture~\ref{c2} is an orthogonality hypergraph $H$ that has some non-adjacent vertices but Criteria~\ref{completion} and~\ref{AdCr} detect and identify their adjacency ``as if'' they were  on a hyperedge in the reconstructed hypergraph $H'$.  Consequently these two hypergraphs $H$ and $H'$ are not the same. Note that, up to permutations of rows and columns, their Travis matrices are the same. Since our dimension is 3, we must have at least three vertices $a$, $b$, and $c$, which are not adjacent in $H$, but whenever one is assigned true by a two-valued state, the other two  have to be assigned false. Moreover, there must not be a two-valued state that simultaneously assigns $a$, $b$, and $c$ false.
	
	Suppose that $G$ is an $(a,b)$-TIFS gadget such that whenever $a$ is assigned true by a two-valued state, $b$ has to be false. As mentioned earlier, the reverse is also true; that is, whenever $b$ is assigned true by a two-valued state, $a$ has to be false. This means that $a$ and $b$ cannot be assigned true at the same time (but they can both be false). Therefore, if we make a triangle, using a true implies false set (TIFS) as its edges, then we have the desired property that whenever one end is true, the other two are false. However, this hypergraph might have several two-valued states that assign false to all its three ends $a$, $b$ and $c$.  Figure~\ref{layer-graph} depicted this layer hypergraph, with $a=a_1 b_3$, $b=a_2 b_1$ and $c=a_3 b_2$.

	\begin{figure}
		\begin{center}
			\begin{tikzpicture}  [scale=1]
				
				\tikzstyle{every path}=[line width=1pt]
				
				\newdimen\ms
				\ms=0.1cm
				\tikzstyle{s1}=[color=black,fill,rectangle,inner sep=3]
				\tikzstyle{c1}=[draw=gray,fill=white,circle,inner sep={\ms/1}]
				

				\coordinate (a1) at (0,0);
				\coordinate (a2) at (3,0);
				\coordinate (a3) at (1.5,2.5981);

				
				
				\draw [color=olive, ->, snake=coil,segment aspect=0,thick] (a1) to node[below] {$G_1$} (a2);
				\draw [color=cyan, ->, snake=coil,segment aspect=0,thick] (a2) to node[right] {$ \; \;  G_2$} (a3);
				\draw [color=orange, ->, snake=coil,segment aspect=0,thick] (a3) to node[left] {$G_3 \; \; $  } (a1);
				
				
				\draw (a1) coordinate[c1,label=below left:$a_{1} b_{3} $];
				\draw (a2) coordinate[c1,label=below right:$a_2 b_1$];
				\draw (a3) coordinate[c1,label=above:$a_3 b_2$];
			\end{tikzpicture}
		\end{center}
		
		\caption{\label{layer-graph}
			A layer hypergraph serving as a quasi-block for the construction of a counterexample to Conjecture \ref{c2}. For $i=1,2,3$, the snake-shaped edges are distinct copies of an arbitrary separable $(a_i , b_i )$-TIFS gadget $G$ such that whenever $a_i$ is assigned true by a two-valued state, $b_i$ has to be false: their respective ends have been ``cyclically folded'' on each other, eg., $a_1$ from $G_1$ is identified with $b_3$ from $G_3$. Note that vertices $a_1 b_3$, $a_2 b_1$ and $a_3 b_2$ are not adjacent.}
	\end{figure}
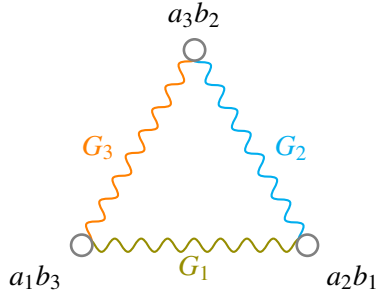
	
	For an $(a_i ,b_i)$-TIFS gadget $G$, Criteria \ref{completion} and \ref{AdCr} do not detect a hyperedge containing $a_1 b_3$, $a_2 b_1$ and $a_3 b_2$ from the two-valued states of a hypergraph such as in Figure~\ref{layer-graph}. This is because usually there are some two-valued states that assign false to all of these three vertices. Therefore, the layer hypergraph of Figure~\ref{layer-graph} must become a part of a larger hypergraph so that in every two-valued state, exactly one of the vertices $a$, $b$ or $c$ be assigned true.
	
	To do so, one possibility is to use three copies of the layer hypergraph of Figure~\ref{layer-graph} and use three extra contexts to bind them together; a configuration drawn in Figure~\ref{TIFS-non-Rec}. When we use $G$ as an $(a_i , b_i)$-TIFS gadget the resulting hypergraph of this construction is denoted by $B(G)$.
	
	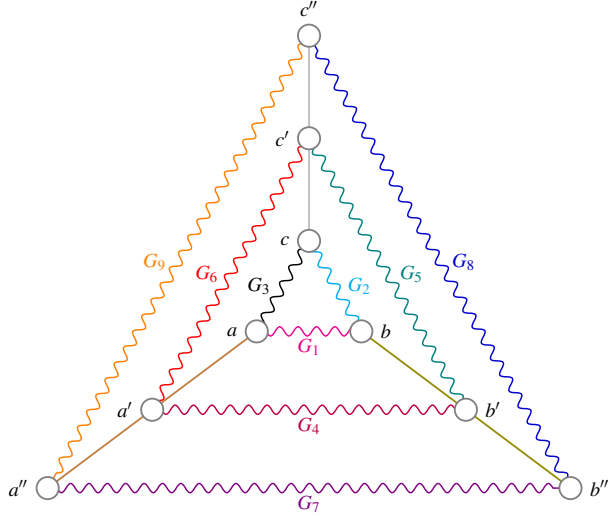
\begin{figure}
		\begin{center}
			\resizebox{0.5\textwidth}{!}{%
				\begin{tikzpicture} [scale=1]
					
					\tikzstyle{every path}=[line width=1pt]
					
					\newdimen\ms
					\ms=0.15cm
					\tikzstyle{s1}=[color=black,fill,rectangle,inner sep=3]
					\tikzstyle{c1}=[draw=gray,fill=white,circle,inner sep={\ms/1}]
					

					\coordinate (a1) at (0,0);
					\coordinate (a2) at (2,0);
					\coordinate (a3) at (1,1.732);
					\coordinate (a4) at (-2,-1.5);
					\coordinate (a5) at (4,-1.5);
					\coordinate (a6) at (1,3.6961);
					\coordinate (a7) at (-4,-3);
					\coordinate (a8) at (6,-3);
					\coordinate (a9) at (1,5.6602);
					
					\draw [color=magenta, ->,snake=coil,segment aspect=0,thick] (a1) to node[below] {$G_1$} (a2);
					\draw [color=cyan, ->,snake=coil,segment aspect=0,thick] (a2) to node[right] {$ \textnormal{       } G_2$} (a3);
					\draw [color=black, ->,snake=coil,segment aspect=0,thick] (a3) to node[left] {$G_3 \textnormal{       }  $  } (a1);
					\draw [color=purple, ->,snake=coil,segment aspect=0,thick] (a4) to node[below] {$G_4$} (a5);
					\draw [color=teal, ->,snake=coil,segment aspect=0,thick] (a5) to node[right] {$ \textnormal{       } G_5$} (a6);
					\draw [color=red, ->,snake=coil,segment aspect=0,thick] (a6) to node[left] {$G_6 \textnormal{       }  $  } (a4);
					\draw [color=violet, ->,snake=coil,segment aspect=0,thick] (a7) to node[below] {$G_7$} (a8);
					\draw [color=midblue, ->,snake=coil,segment aspect=0,thick] (a8) to node[right] {$ \textnormal{       } G_8$} (a9);
					\draw [color=orange, ->,snake=coil,segment aspect=0,thick] (a9) to node[left] {$G_9 \textnormal{       }  $  } (a7);

					\draw [color=brown] (a1) -- (a7);
					\draw [color=olive] (a2) -- (a8);
					\draw [color=lightgray] (a3) -- (a9);
					
					
					\draw (a1) coordinate[c1,label=left:$a$];
					\draw (a2) coordinate[c1,label=right:$b$];
					\draw (a3) coordinate[c1,label=left:$c$];
					\draw (a4) coordinate[c1,label=left:$a'$];
					\draw (a5) coordinate[c1,label=right:$b'$];
					\draw (a6) coordinate[c1,label=left:$c'$];
					\draw (a7) coordinate[c1,label=left:$a''$];
					\draw (a8) coordinate[c1,label=right:$b''$];
					\draw (a9) coordinate[c1,label=above:$c''$];		
				\end{tikzpicture}
			}
		\end{center}
		\caption{\label{TIFS-non-Rec}
			A hypergraph $B(G)$ depicting a counterexample to Conjecture \ref{c2}. For $i=1,\ldots ,9$, the snake-shaped curves indicate different copies of an $(a_i , b_i )$-TIFS gadget $G$ with their terminals suitably identified, that is, $a_1$ from $G_1$ is identified with $b_3$ from $G_3$ as the vertex $a$. Straight lines are ordinary hyperedges, i.~e., vertices $a$, $a'$ and $a''$ are on a context, drawn in brown.}
	\end{figure}
	
	Since $G$ is a TIFS gadget, in every two-valued state of $B(G)$, exactly one of $a$, $b$ or $c$ has to be assigned true while the other two have to be false. This is because $\{a,a',a''\}$, $\{b,b',b''\}$, and $\{c,c',c''\}$ are contexts and need to have exactly one true value in every two-valued states. Consequently, if for example $b$ and $c$ are assigned false by a two-valued state $t_r$, then one of $b'$ and $b''$, and one of $c'$ and $c''$ have to be assigned true by $t_r$. Without loss of generality, let $b'$ be the true one and $b''$ be false. Then, since $G_4$ and $G_5$ are TIFS gadgets, it can be inferred that $a'$ and $c'$ have to be assigned false by $t_r$. Therefore, $c''$ has also to be assigned true because $c$ and $c'$ are false. Again, since $G_9$ is a TIFS gadget, $a''$ has to be false. Now, $a$ has to be assigned true by $t_r$ because $a'$ and $a''$ are both false.
	
	Therefore, for nonadjacent vertices $a$, $b$ and $c$ we have that exactly one of them has to be assigned true, and the other two are false. This is just like if these three vertices are on the same context. It means that if one tries to reconstruct $B(G)$ from its table of two-valued states using criteria \ref{completion} and \ref{AdCr},  extra hyperedges of $\{a,b,c\}$, $\{a',b',c'\}$ and $\{a'',b'',c''\}$ are found. This implies that $B(G)$ is not reconstructable using these criteria.
	
	One question that arises in this regard is the following: suppose we know that $G$ is separable, then does this imply that $B(G)$ is also separable? This is not too hard to answer, and it is always ``yes'' if it does not contain a true implies true set (TITS) gadget.
	
	To proceed, we need one more thing. A function $f:S\longrightarrow T$ defined on a subset $S\subsetneq X$ is said to be \emph{lifted} to $\tilde{f}:X\longrightarrow T$ if $\tilde{f}(a)=f(a)$ for each $a\in S$. It must be mentioned that when $f$ possesses a property, like if $f$ is a proper coloring~\cite{ALBERTSON1998189} or a two-valued state, there might not always be a lift with the same property.
	
	\begin{lemma}\label{Sep-lem}
		Let $G$ be a separable unitary TIFS gadget which does not contain a TITS gadget and $\omega ([G]_2)=3$. Then $B(G)$ is also separable.
	\end{lemma}
	
	\begin{proof}
		Since $G$ is unitary, there is no vertex that has to be assigned 0 by all the two-valued states of $G$. On the other hand, there is no vertex that is given 1 by all the two-valued states of $G$ because else it must lie on at least one hyperedge with two other vertices, those that have to be always assigned 0, a contradiction to separability (and being unitary) of $G$. Thus, it can be inferred that for a vertex $u$ of $G$, there are states $\varphi$ and $\psi$ such that $\varphi (u)=0$ and $\psi (u)=1$.
		
		For an $(a ,b )$-TIFS gadget $G$, the set of two-valued states are nonempty. Let $s_{1}$ be the set of those states in which $a$ is true, $s_{2}$ be the set of those states that $b$ is true and $s_{3}$ be the set of those states in which both $a$ and $b$ are false. Then these sets $s_{1}$ , $s_{2}$ and $s_{3}$ partition the set of two-valued states of $G$. For an $(a ,b )$-TIFS gadget copy $G_i$, these sets are shown here by $s_{i1}$, $s_{i2}$ and $s_{i3}$. Therefore, if for example a two-valued state is in $s_{1}$, it is a two-valued state of $G$ so that its head, i.~e. $a$, is assigned 1.
		
		To show that $B(G)$ is separable, we show that for any pair of distinct vertices $x$ and $y$, there is a two-valued state of it that gives them different values. There are the following cases:
		
		\begin{itemize}
			\item[Case 1.] \emph{$x$ and $y$ belong to the same copy of $G$, say $G_i$.} \ Because of the symmetry, we can assume without loss of generality that $i=1$.  Since $G$ is separable, there must be a two-valued state, say $t$ on $G$, such that $t(x)\neq t(y)$. We know that $t\in s_{1}\cup s_{2}\cup s_{3}$. If $t\in s_l$ for $l=1,2,3$, then there is a two-valued state for the underlying hypergraph of $B(G)$ in Section~\ref{sec:underlying} such that it agrees with the values of $t(a)$ and $t(b)$. Now, using this two-valued state we define $\Tilde{t}$ for the end vertices $a$, $b$, $c$, $a'$, $b'$ ,$c'$, $a''$, $b''$ and $c''$. Then using appropriate two-valued states of $G$, we can find a suitable two-valued state for internal vertices of $G_2,\ldots,G_9$. Therefore, there are two-valued states such as $\Tilde{t}$ of $B(G)$ which is a lifting for $t$ and  $\Tilde{t}(x) \neq \Tilde{t}(y)$.

			\item[Case 2.] \emph{$x$ and $y$ lie on different copies of $G$, say $G_i$ and $G_j$ respectively.} Then there are three other cases.
			\begin{itemize}
				\item[Case 2.1.] \emph{$G_i$ and $G_j$ lie on the same layer of $B(G)$.} Without loss of generality, suppose that it is the layer consisting of the vertices $a$, $b$, and $c$. Then $G_i$ and $G_j$ have a common vertex that, again without loss of generality, we can assume it is $a$. Let $G_i$'s head and tail be $a$ and $b$, and $G_j$'s head and tail be $c$ and $a$, respectively [see Figure \ref{fig-proof-sep-lem} (a)]. Therefore, every two-valued state of the induced subhypergraph $G_i \cup G_j$ in $B(G)$ is a member of $s_{i1} \cup s_{j2}$, $s_{i2} \cup s_{j3}$ or $s_{i3}\cup s_{j1}$. Suppose on contrary that $x$ and $y$ receive the same value by all two-valued states of $B(G)$. Then, since there is a two-valued state of $B(G)$ that assigns 1 to $x$, at least one of the following statements holds:
				
				\begin{itemize}
					\item[1.] If $x$ is assigned 1 by a two-valued state of $s_{i1}$, then $y$ has to be assigned 1 by all two-valued states of $s_{j2}$. This means that $G_j$ (and therefore $G$) is a $(a,y)$-TITS gadget, a contradiction to our assumption.
					\item[2.] If $x$ is assigned 1 by a two-valued state of $s_{i2}$, then $y$ has to be assigned 1 by all two-valued states of $s_{j3}$. Consequently, $x$ cannot be assigned 0 by $s_{i2}$ because else, it can be lifted to the required separation of $x$ and $y$. Hence, $G_i$ (and therefore $G$) is a $(b,x)$-TITS gadget, a contradiction to our assumption.
					\item[3.] If $x$ is assigned 1 by a two-valued state of $s_{i3}$, then $y$ has to be assigned 1 by all two-valued states of $s_{j1}$. This means that $G_j$ (and therefore $G$) is a $(c,y)$-TITS gadget, again a contradiction to our assumption.
				\end{itemize}

				\item[Case 2.2] \emph{$G_i$ and $G_j$ lie on different layers of $B(G)$, and both ends of $G_i$ and $G_j$ lie on the same contexts of $B(G)$.} Without loss of generality suppose that $G_i$ has $a$ and $b$ and $G_j$ has $a'$ and $b'$ as their heads and tails, respectively [see Figure \ref{fig-proof-sep-lem} (b)]. Therefore, every two-valued state of the induced subhypergraph $G_i \cup G_j$ in $B(G)$ is a member of $s_{i1} \cup s_{j2}\cup s_{j3}$, $s_{i2} \cup s_{j1}\cup s_{j3}$ or $s_{i3}\cup s_{j1}\cup s_{j2}$ (or $s_{j1} \cup s_{i2}\cup s_{i3}$, $s_{j2} \cup s_{i1}\cup s_{i3}$ or $s_{j3}\cup s_{i1}\cup s_{i2}$ which are completely similar). Suppose on contrary that $x$ and $y$ receive the same value by all two-valued states of $B(G)$. Then, since there is a two-valued state of $B(G)$ that assigns 1 to $x$, at least one of the following statements holds:
				
				\begin{itemize}
					\item[1.] If $x$ is assigned 1 by a two-valued state of $s_{i1}$, then $y$ has to be assigned 1 by all two-valued states of $s_{j2}\cup s_{j3}$. This means that $G_j$ (and therefore $G$) is a $(b',y)$-TITS gadget, a contradiction to our assumption.
					\item[2.] If $x$ is assigned 1 by a two-valued state of $s_{i2}$, then $y$ has to be assigned 1 by all two-valued states of $s_{j1}\cup s_{j3}$. Consequently, $x$ cannot be assigned 0 by a $s_{i2}$ because else, it can be lifted to the required separation of $x$ and $y$. Hence, $G_i$ (and therefore $G$) is a $(b,x)$-TITS gadget, a contradiction to our assumption.
					\item[3.] If $x$ is assigned 1 by a two-valued state of $s_{i3}$, then $y$ has to be assigned 1 by all two-valued states of $s_{j1}\cup s_{j2}$. This means that $G_j$ (and therefore $G$) is a $(b',y)$-TITS gadget (and also a $(c',y)$-TITS gadget), again a contradiction to our assumption.
				\end{itemize}
				\item[Case 2.3] \emph{$G_i$ and $G_j$ lie on different layers of $B(G)$, but only one end from $G_i$ and one end from $G_j$ lie on the same context.} Again without loss of generality, suppose that $G_i$ has $a$ and $b$ and $G_j$ has $b'$ and $c'$ as their heads and tails, respectively [see Figure \ref{fig-proof-sep-lem} (c)]. Therefore, every two-valued state of the induced subhypergraph $G_i \cup G_j$ in $B(G)$ is a member of $s_{i1} \cup s_{j1}\cup s_{j2}\cup s_{j3}$, $s_{i2} \cup s_{j2}\cup s_{j3}$, or $s_{i3}\cup s_{j1}\cup s_{j2} \cup s_{j3}$ (or $s_{j1} \cup s_{i1}\cup s_{i2}\cup s_{i3}$, $s_{j2} \cup s_{i2}\cup s_{i3}$, or $s_{j3}\cup s_{i1}\cup s_{i2} \cup s_{i3}$ which can be treated similarly). Suppose on contrary that $x$ and $y$ receive the same value by all two-valued states of $B(G)$. Then, since there is a two-valued state of $B(G)$ that assigns 1 to $x$, at least one of the following statements holds:
				
				\begin{itemize}
					\item[1.] If $x$ is assigned 1 by a two-valued state of $s_{i1}$, then $y$ has to be assigned 1 by all two-valued states of $s_{j1}\cup s_{j2}\cup s_{j3}$, a contradiction to the fact that no vertex of $G$ can be assigned 1 by all the two-valued states.
					\item[2.] If $x$ is assigned 1 by a two-valued state of $s_{i2}$, then $y$ has to be assigned 1 by all two-valued states of $s_{j2}\cup s_{j3}$. Consequently, $G_j$ (and therefore $G$) is a $(c',y)$-TITS gadget, a contradiction to our assumption.
					\item[3.] If $x$ is assigned 1 by a two-valued state of $s_{i3}$, then $y$ has to be assigned 1 by all two-valued states of $s_{j1}\cup s_{j2}\cup s_{j3}$, a contradiction to the fact that no vertex of $G$ can be assigned 1 by all the two-valued states.
				\end{itemize}
				
			\end{itemize}
		\end{itemize}

		We showed that in any case, $x$ and $y$ can be separated by a two-valued state of $B(G)$ (or else there is a contradiction) which concludes the proof.
	\end{proof}

	\begin{figure}
		\begin{center}
			\begin{tabular}{ c c c c c }
				
				\begin{tikzpicture} [scale=0.44]
					
					\tikzstyle{every path}=[line width=1pt]
					
					\newdimen\ms
					\ms=0.15cm
					\tikzstyle{s1}=[color=black,fill,rectangle,inner sep=3]
					\tikzstyle{c1}=[draw=gray,fill=white,circle,inner sep={1}]
					

					\coordinate (a1) at (0,0);
					\coordinate (a2) at (2,0);
					\coordinate (a3) at (1,1.732);
					\coordinate (a4) at (-2,-1.5);
					\coordinate (a5) at (4,-1.5);
					\coordinate (a6) at (1,3.6961);
					\coordinate (a7) at (-4,-3);
					\coordinate (a8) at (6,-3);
					\coordinate (a9) at (1,5.6602);

					\draw [color=lightgray, ->,snake=coil,segment aspect=0,thick] (a1) to node[below] { } (a2);
					\draw [color=lightgray, ->,snake=coil,segment aspect=0,thick] (a2) to node[right] { } (a3);
					\draw [color=lightgray, ->,snake=coil,segment aspect=0,thick] (a3) to node[left] {  } (a1);
					\draw [color=lightgray, ->,snake=coil,segment aspect=0,thick] (a4) to node[below] { } (a5);
					\draw [color=lightgray, ->,snake=coil,segment aspect=0,thick] (a5) to node[right] { } (a6);
					\draw [color=lightgray, ->,snake=coil,segment aspect=0,thick] (a6) to node[left] {  } (a4);
					\draw [color=violet, ->,snake=coil,segment aspect=0,thick] (a7) to node[below] {$G_i$} (a8);
					\draw [color=lightgray, ->,snake=coil,segment aspect=0,thick] (a8) to node[right] { } (a9);
					\draw [color=orange, ->,snake=coil,segment aspect=0,thick] (a9) to node[left] {$G_j \textnormal{       }  $  } (a7);

					\draw [color=lightgray] (a1) -- (a7);
					\draw [color=lightgray] (a2) -- (a8);
					\draw [color=lightgray] (a3) -- (a9);
					
					
					\draw (a1) coordinate[c1];
					\draw (a2) coordinate[c1];
					\draw (a3) coordinate[c1];
					\draw (a4) coordinate[c1];
					\draw (a5) coordinate[c1];
					\draw (a6) coordinate[c1];
					\draw (a7) coordinate[c1,label=below:$a$];
					\draw (a8) coordinate[c1,label=below:$b$];
					\draw (a9) coordinate[c1,label=above:$c$];		
				\end{tikzpicture}
				
				& \hspace{0.6 cm} &
				
				\begin{tikzpicture} [scale=0.44]
					
					\tikzstyle{every path}=[line width=1pt]
					
					\newdimen\ms
					\ms=0.15cm
					\tikzstyle{s1}=[color=black,fill,rectangle,inner sep=3]
					\tikzstyle{c1}=[draw=gray,fill=white,circle,inner sep={1}]
					

					\coordinate (a1) at (0,0);
					\coordinate (a2) at (2,0);
					\coordinate (a3) at (1,1.732);
					\coordinate (a4) at (-2,-1.5);
					\coordinate (a5) at (4,-1.5);
					\coordinate (a6) at (1,3.6961);
					\coordinate (a7) at (-4,-3);
					\coordinate (a8) at (6,-3);
					\coordinate (a9) at (1,5.6602);

					\draw [color=lightgray, ->,snake=coil,segment aspect=0,thick] (a1) to node[below] { } (a2);
					\draw [color=lightgray, ->,snake=coil,segment aspect=0,thick] (a2) to node[right] { } (a3);
					\draw [color=lightgray, ->,snake=coil,segment aspect=0,thick] (a3) to node[left] {  } (a1);
					\draw [color=orange, ->,snake=coil,segment aspect=0,thick] (a4) to node[above] {$G_j$ } (a5);
					\draw [color=lightgray, ->,snake=coil,segment aspect=0,thick] (a5) to node[right] { } (a6);
					\draw [color=lightgray, ->,snake=coil,segment aspect=0,thick] (a6) to node[left] {  } (a4);
					\draw [color=violet, ->,snake=coil,segment aspect=0,thick] (a7) to node[below] {$G_i$} (a8);
					\draw [color=lightgray, ->,snake=coil,segment aspect=0,thick] (a8) to node[right] { } (a9);
					\draw [color=lightgray, ->,snake=coil,segment aspect=0,thick] (a9) to node[left] {  } (a7);

					\draw [color=lightgray] (a1) -- (a7);
					\draw [color=lightgray] (a2) -- (a8);
					\draw [color=lightgray] (a3) -- (a9);
					
					
					\draw (a1) coordinate[c1];
					\draw (a2) coordinate[c1];
					\draw (a3) coordinate[c1];
					\draw (a4) coordinate[c1,label=below:$a'$];
					\draw (a5) coordinate[c1,label=below:$b'$];
					\draw (a6) coordinate[c1];
					\draw (a7) coordinate[c1,label=below:$a$];
					\draw (a8) coordinate[c1,label=below:$b$];
					\draw (a9) coordinate[c1,label=above:$c$];		
				\end{tikzpicture}
				
				& \hspace{0.6 cm} &
				
				\begin{tikzpicture} [scale=0.44]
					
					\tikzstyle{every path}=[line width=1pt]
					
					\newdimen\ms
					\ms=0.15cm
					\tikzstyle{s1}=[color=black,fill,rectangle,inner sep=3]
					\tikzstyle{c1}=[draw=gray,fill=white,circle,inner sep={1}]
					

					\coordinate (a1) at (0,0);
					\coordinate (a2) at (2,0);
					\coordinate (a3) at (1,1.732);
					\coordinate (a4) at (-2,-1.5);
					\coordinate (a5) at (4,-1.5);
					\coordinate (a6) at (1,3.6961);
					\coordinate (a7) at (-4,-3);
					\coordinate (a8) at (6,-3);
					\coordinate (a9) at (1,5.6602);

					\draw [color=lightgray, ->,snake=coil,segment aspect=0,thick] (a1) to node[below] { } (a2);
					\draw [color=lightgray, ->,snake=coil,segment aspect=0,thick] (a2) to node[right] { } (a3);
					\draw [color=lightgray, ->,snake=coil,segment aspect=0,thick] (a3) to node[left] {  } (a1);
					\draw [color=orange, ->,snake=coil,segment aspect=0,thick] (a4) to node[below] { $G_j$ } (a5);
					\draw [color=lightgray, ->,snake=coil,segment aspect=0,thick] (a5) to node[right] { } (a6);
					\draw [color=lightgray, ->,snake=coil,segment aspect=0,thick] (a6) to node[left] {  } (a4);
					\draw [color=lightgray, ->,snake=coil,segment aspect=0,thick] (a7) to node[below] { } (a8);
					\draw [color=lightgray, ->,snake=coil,segment aspect=0,thick] (a8) to node[right] { } (a9);
					\draw [color=violet, ->,snake=coil,segment aspect=0,thick] (a9) to node[left] {$G_i \textnormal{       }  $  } (a7);

					\draw [color=lightgray] (a1) -- (a7);
					\draw [color=lightgray] (a2) -- (a8);
					\draw [color=lightgray] (a3) -- (a9);
					
					
					\draw (a1) coordinate[c1];
					\draw (a2) coordinate[c1];
					\draw (a3) coordinate[c1];
					\draw (a4) coordinate[c1,label=below:$b'$];
					\draw (a5) coordinate[c1,label=below:$c'$];
					\draw (a6) coordinate[c1];
					\draw (a7) coordinate[c1,label=below:$b$];
					\draw (a8) coordinate[c1,label=below:$c$];
					\draw (a9) coordinate[c1,label=above:$a$];		
				\end{tikzpicture}
				
				\\
				
				(a)&&(b)&&(c)	
			\end{tabular}
		\end{center}
		\caption{\label{fig-proof-sep-lem}
			Illustrations of cases 2.1 (a), 2.2 (b) and 2.3 (c) in the proof of Lemma \ref{Sep-lem}.}
	\end{figure}
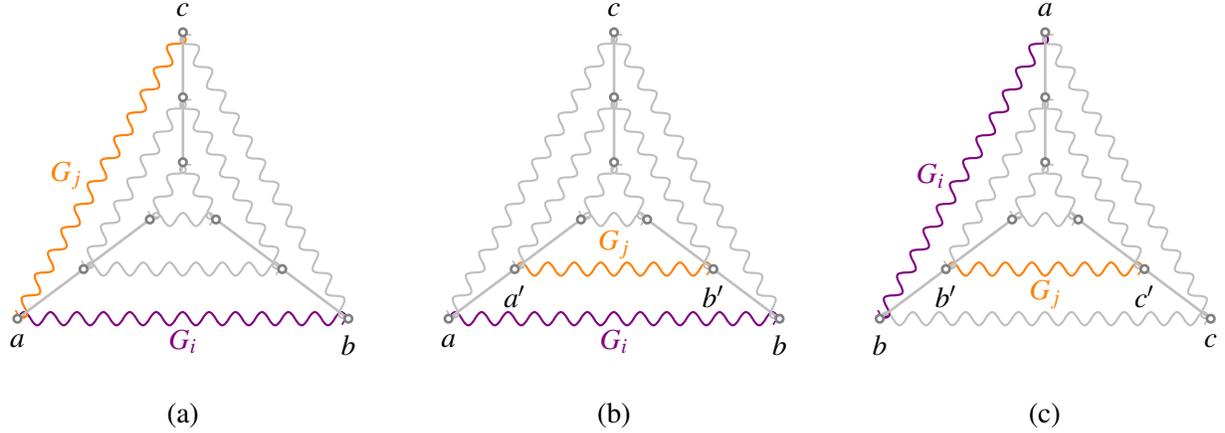
	
	We can go further by calculating the number of two-valued states of $B(G)$ based on $G$'s. Suppose that $G$ is an $(a,b)$-TIFS and has, respectively, $n_a$, $n_b$, and $n_n$ two-valued states that give $a$ true, $b$ true, and none of $a$ and $b$ true. In other words, $n_a =\vert s_1 \vert$, $n_b = \vert s_2 \vert$, and $n_n = \vert s_3 \vert$. Then, by the elementary methods of counting, the number of two-valued states of $B(G)$ is
	\begin{equation}\label{nTS}
		nTS(B(H))= (3+2+1)\cdot n_a^3 \cdot n_b^3 \cdot n_n^3 .
	\end{equation}
	
	For the sake of an example take the Specker bug $G$  discussed in the Appendix Section~\ref{Specker-Bug} of Appendix~\ref{AppendixA}. From its Travis matrix \ref{Specker-Bug-Travice} we know that $n_a = 3$, $n_b =3$ and $n_n=8$. Therefore, Formula \ref{nTS} implies that $B(G)$, which is a separable hypergraph on 108 vertices and 66 contexts, has $$6\cdot 3^3 \cdot 3^3\cdot 8^3=2,239,488$$ two-valued states, a number that can easily be checked via an ordinary computer.
	
	However, it is not difficult to show that this hypergraph $B(G)$, when $G$ is the Specker bug, does not meet our requirements. This is because the two ends of the Specker bug cannot be orthogonal in the 3-space \cite{Cabello-1996-diss}. It can also be discussed using graph theoretical terminology; one orthogonality hypergraph cannot have a cycle of length 4 because else any pair of antipodal vertices of the cycle of length 4 have to be colinear. Therefore, even if $B(G)$ is an orthogonality hypergraph, the reconstructed hypergraph with an extra context of $\{a,b,c\}$ is certainly not.
	
	We have to find a TIFS gadget $H$, other than the Specker bug, in which the distance between its two ends is not 3 (so the two ends can be orthogonal in the reconstructed hypergraph). A candidate for such a hypergraph is shown in Figure~\ref{Baba-Taher}, which is a TIFS on 43 vertices whose end points (say $a$ and $b$) are far enough, so that it is not only separable, but also probably a FOR. This hypergraph has 2589 two-valued states, 45 of which assign $a$ true, 504 give $b$ true and 2040 give both $a$ and $b$ false. In other words, for the hypergraph $H$ we have $n_a=45$, $n_b=504$ and $n_n=2048$.
	
	If we construct $B(H)$, then Lemma \ref{Sep-lem} implies that the resulting hypergraph on 378 vertices is separable and there would be enough space for the three end vertices to be perpendicular.
	
	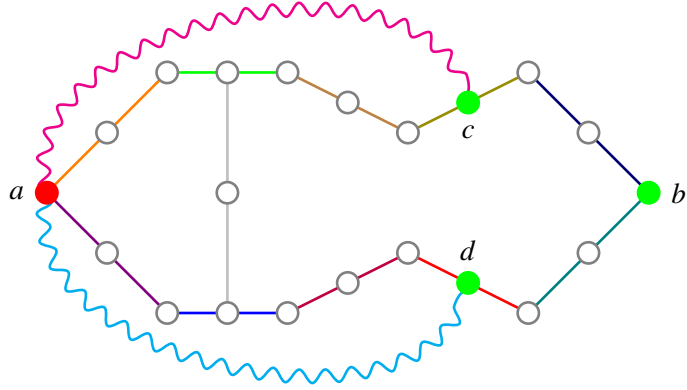
\begin{figure}
		\begin{center}
			\begin{tikzpicture}  [scale=0.8]
				
				\newdimen\ms
				\ms=0.1cm
				
				\tikzstyle{every path}=[line width=1pt]
				
				\tikzstyle{c1}=[draw=gray,fill=white,circle,inner sep={\ms/1}]
				\tikzstyle{c2}=[color=blue,fill,circle,inner sep={\ms/8},minimum size=2*\ms]
				\tikzstyle{c3}=[color=red,fill,circle,inner sep={\ms/8},minimum size=2*\ms]
				
				\newdimen\R
				\R=30mm     
				
				
				
				%

				\coordinate (a1) at (0,0);
				\coordinate (a2) at (1,1);
				\coordinate (a3) at (2,2);
				\coordinate (a4) at (3,2);
				\coordinate (a5) at (4,2);
				\coordinate (a6) at (5,1.5);
				\coordinate (a7) at (6,1);
				\coordinate (a8) at (7,1.5);
				\coordinate (a9) at (8,2);
				\coordinate (a10) at (9,1);
				\coordinate (a11) at (10,0);
				\coordinate (a12) at (9,-1);
				\coordinate (a13) at (8,-2);
				\coordinate (a14) at (7,-1.5);
				\coordinate (a15) at (6,-1);
				\coordinate (a16) at (5,-1.5);
				\coordinate (a17) at (4,-2);
				\coordinate (a18) at (3,-2);
				\coordinate (a19) at (2,-2);
				\coordinate (a20) at (1,-1);
				\coordinate (a21) at (3,0);
				
				\draw [color=magenta, decoration = snake,decorate] (a1) to [bend left=90]  (a8);
				\draw [color=cyan, decoration = snake,decorate] (a1) to [bend right=90]  (a14);
				
				\draw [color=orange] (a1) -- (a3);
				\draw [color=green] (a3) -- (a5);
				\draw [color=brown] (a5) -- (a7);
				\draw [color=olive] (a7) -- (a9);
				\draw [color=darkblue] (a9) -- (a11);
				\draw [color=teal] (a11) -- (a13);
				\draw [color=red] (a13) -- (a15);
				\draw [color=purple] (a15) -- (a17);
				\draw [color=blue] (a17) -- (a19);
				\draw [color=violet] (a19) -- (a1);
				\draw [color=lightgray] (a4) -- (a18);
				
				\draw (a1) coordinate[c1,draw=red,fill=red,label=left:$a$];
				\draw (a2) coordinate[c1];
				\draw (a3) coordinate[c1];
				\draw (a4) coordinate[c1];
				\draw (a5) coordinate[c1];
				\draw (a6) coordinate[c1];
				\draw (a7) coordinate[c1];
				\draw (a8) coordinate[c1,draw=green,fill=green,label=below:$c$];
				\draw (a9) coordinate[c1];
				\draw (a10) coordinate[c1];
				\draw (a11) coordinate[c1,draw=green,fill=green,label=right:$b$];
				\draw (a12) coordinate[c1];
				\draw (a13) coordinate[c1];
				\draw (a14) coordinate[c1,draw=green,fill=green,label=above:$d$];
				\draw (a15) coordinate[c1];
				\draw (a16) coordinate[c1];
				\draw (a17) coordinate[c1];
				\draw (a18) coordinate[c1];
				\draw (a19) coordinate[c1];
				\draw (a20) coordinate[c1];
				\draw (a21) coordinate[c1];		
			\end{tikzpicture}
		\end{center}
		\caption{\label{Baba-Taher}
			An $(a,b)$-TIFS gadget whose distance between its two terminal points $a$ and $b$ is at least five contexts. The snake-like decorated curves indicate Specker bugs, so this hypergraph has 43 vertices.
		}
	\end{figure}
	
	The hypergraph $B(H)$ is on 378 vertices and 228 contexts, and by Formula \ref{nTS} it has $$6\cdot 45^3\cdot 504^3\cdot 2040 ^3=594,252,343,817,330,688,000,000$$ two-valued states. This huge number makes it hard to quickly check by using an ordinary computer.
	
	Another pertinent problem is to show that these hypergraphs---namely $B(H)$ and also its counterpart, the reconstructed hypergraph $B(H)'$ from $B(H)$'s table of two-valued states---have a faithful orthogonal representation; and to enumerate an explicit example of such a representation.
	
	However, it seems that Criterion \ref{orthogonality} is independent of Criteria \ref{completion} and \ref{AdCr}, having an example such as $B(H)$ on 378 vertices raises the possibility that Conjecture \ref{c2} is false for separable hypergraphs that are not perfectly-separable.
	
	One challenge is to find either TIFS that allows for FORs---that is, vertex labellings by vectors---with end points that are orthogonal (that is, their relative angle is $\pi/2$) and at the same time have a separating set of two valued states; or a proof of nonexistence thereof.  Note that since, unlike TIFS,  TITS gadgets in general perform asymmetric, it is not possible to employ a serial composition strategy similar to the one of Kochen and Specker~\cite{kochen1} for a construction of their $\Gamma_2$:  to concatenate a couple of TITS with a single TIFS at their respective end points and thereby to obtain a TIFS with a ``larger aperture'': the TITS relationship ``one-implies-one'' in general works only one way (with respect to exchange of the end points), and not the other way around. A related challenge is to construct a symmetrical TITS, realizing the ``one-implies-one'' relation at both end points.
	
	\section{Two-valued states vs coloring}\label{color-conj}
	
	\subsection{Coloring vs independent partitions}
	A hypergraph [whose vertices all lie on at least one hyperedge of size $\omega([H]_2 )=n$] is called to have ``an $n$-partition system'' if there is a vertex partition $\mathcal{S}=\{S_1 , \ldots , S_n \}$ of $H$ (into exactly $n$ cells) with the following properties:
	\begin{itemize}
		\item[1.] whenever $v,w \in S_i$ for an $i=1,\ldots , n$, we have that $v$ and $w$ are not adjacent, and
		\item[2.] for each $i=1,\ldots , n$ and every $v \in V(G)$, there is a vertex $w \in S_i$ such that either $v =w$ or $v$ and $w$ are adjacent.
	\end{itemize}
	
	When a hypergraph $H$ has an $n$-partition system, we might simply say that $H$ is $n$-partitionable. While $\mathcal{S}$ is a partition, every $S_i$ is non-empty and $\bigcup_{i=1}^{n} S_i = V(G)$. With Property 1, we can be sure that if we assign the true value to all the vertices in $S_i$ and the false value to the rest of vertices, then no two true vertices are adjacent, and consequently, every context has at most one true valued vertex. Moreover, Property 2 assures us that there is no context without a true-valued vertex. (Using graph theoretical terminology, these two properties mean that every $S_i$ is an \emph{independent dominating set} for $H$.)  Therefore, an $n$-partition system actually induces $n$ two-valued states on $H$.
	
	We have the following theorem.
	
	\begin{theorem}\label{separable}
		A hypergraph $H$ is $n$-colorable if and only if it is $n$-partitionable.
	\end{theorem}
	\begin{proof}
		It is evident that if $\sigma$ is a proper vertex coloring of $G$ with $\{1,\ldots , n\}$, we can easily find $\mathcal{S}=\{S_1 , \ldots , S_k \}$ by putting $$S_i =\{v\in V(G)\; : \; \sigma(v)=i \}.$$ It is also clear that Property 1 holds because $\sigma$ is a proper coloring. Moreover, we also have Property 2 because of the assumption we made that, in $[H]_2$, every vertex is on a maximal clique of size $n$.
		
		To prove the converse, suppose that we have a partition $\mathcal{S}=\{S_1 , \ldots , S_n \}$ of vertices of $H$ which satisfies Properties 1 and 2 above. Define, $\sigma: V(G) \longrightarrow \{1,\ldots , n\}$ such that
		
		\begin{center}
			$\sigma(v)=i$ if $v\in S_i.$
		\end{center}
		
		While Property 2 implies that every $S_i$ is non-empty, for $i=1\ldots , n$, it also shows that every hyperedge contains a vertex $v$ such that $\sigma (v)=i$. In other words, every color $i=1,\ldots,k$ is used in each hyperedge. Furthermore, Property 1 implies that every $S_i$ is an independent set. Hence $\sigma$ is a proper coloring of $H$ and consequently, $H$ is $n$-colorable.
	\end{proof}
	
	As a result, we can say that for each proper $n$-coloring of $H$ we have $n$ different two-valued states on vertices of $H$. Conversely we can construct exactly $n!$ proper $n$-colorings for $H$ from an available $n$-partition system on vertices of $H$. 
	Therefore, the following corollary (Corollary~\ref{c1}) is a consequence of Theorem~\ref{separable}.
	
	\begin{corollary}\label{c1} The following statements are equivalent:
		
		\begin{itemize}
			\item[(i)]
			The hypergraph $H$ is semi-perfect, i.e., its chromatic number and 2-section clique number are equal.
			\item[(ii)]
			The set of two-valued states contains $n$ members which correspond
			to, or induce, a partitioning of all elements of the partition logic;
			the equivalence relation defined by each one of these $n$ states evaluating to $1$ on some element of every context.
			That is, those $n$ states are $1$ on different atoms of every context.
		\end{itemize}
	\end{corollary}
	
	This does not exclude the existence of partition logics which are not semi-perfect.
	Indeed, in general, their chromatic number can exceed their 2-section's clique number.
	A concrete example is
	Greechie's $G_{32}$~\cite[Figure~6, p.~121]{greechie:71} mentioned in Appendix~\ref{2021-chroma-G32}, and depicted in
	Figure~\ref{2020-f-GreechieG32}.

	\subsection{Reconstructing coloring from logical assignments}
	
	From Theorem~\ref{separable} we know that when there is an $n$-coloring for an orthogonality hypergraph $H$, there are $n$ two-valued states corresponding to it so that they induce a partition logic. In other words, there are $n$ rows in the Travis matrix of $H$ such that when one of them assigns 1 to a vertex $u$, the rest of them assign 0 to $u$. Consequently, $t_{s_1},\ldots,t_{s_n}$ are the rows of the Travis matrix $T(H)$ corresponding to an $n$-coloring. This yields the $\vert V(H)\vert$--tuple whose entries are one.
	\begin{equation}\label{color-all-two-states}
		\sum_{i=1}^{n} t_{s_i} = \Big(\underbrace{\raisebox{-0pt}{1,1,\ldots,1}}_{\vert V(H)\vert \text{ times}}\Big)=\begin{pmatrix}\mathbf{1}_{1\times \vert V(H)\vert}\end{pmatrix}.
	\end{equation}
	Therefore, when $H$ is $n$-colorable, there is at least one set of $n$ two-valued states that induce a partition logic on $H$.
	
	Algorithm \ref{algorithm1} searches for such states when $T(H)$ is available.
	\begin{enumerate}
		\item[(1)] It takes the Travis matrix $T(H)$ and the clique number $n$, and
		\item[(2)] It gives a list of rows, $A$, from which we can retrieve an $n$ coloring for $H$.
	\end{enumerate}
	
	Variables of this algorithm are as follows:
	\begin{enumerate}
		\item[(1)] $AvailableRows$ which is a list of active rows in $T(H)$, with each such row representing a two-valued state of $H$:
		\item[(2)] $i$, which runs from $1$ to the clique number $n$,
		\item[(3)] $j$, which runs from $1$ to the number of two-valued states $nTS(H)$, which is the number of rows of the Travis matrix; and
		\item[(4)] $RemovedRows$ which is a list of lists, whose $i$th element is the rows of $T(H)$ that become inactive at the $i$th step of filling $A$.
	\end{enumerate}
	
	\def\NoNumber#1{{\def\alglinenumber##1{}\State #1}\addtocounter{ALG@line}{-1}}
	
	\begin{algorithm}
		\caption{Finding an $n$-coloring for $H$  from its set of two-valued states encoded by the Travis matrix}\label{algorithm1}
		\begin{flushleft} $\;$\\ \hspace*{\algorithmicindent}
			\textbf{Input:} $T(H)$, $n$ \Comment{Travis matrix, clique number}\\
			\hspace*{\algorithmicindent} \textbf{Output:} $A$ \Comment{a list of $n$ rows of $T(H)$}
		\end{flushleft}
		\begin{algorithmic}[1]
			
			\State $i\gets 1$ \Comment{start of variable initialization}
			\State $AvailableRows\gets (1,\ldots,nTS(H))$
			\State $A\gets (\;)$
			\State $RemovedRows\gets (\;)$ \Comment{end of variable initialization}
			\NoNumber{ }
			\While{$i\leq n$ and ($i\neq 1$ or $AvailableRows\neq \emptyset$)} \Comment{try all colors}
			
			\If{$AvailableRows=\emptyset$} \Comment{start over again if all two-valued states are exhausted}
			
			\State {\it Append} $RemovedRows[i]$ to $AvailableRows$
			\State {\it Remove} $RemovedRows[i]$ from $RemovedRows$
			\State {\it Remove} $A[i]$ from $A$
			\State $i\gets i-1$
			
			\Else \Comment{try to identify a new color assignment by the next available two-valued state}
			
			\State $j\gets$ first available cell in $AvailableRows$
			\State $A[i] \gets AvailableRows[j]$
			\State {\it Append} $AvailableRows[j]$ to $RemovedRows[i]$
			\State {\it Remove} $AvailableRows[j]$ from $AvailableRows$
			\State $i\gets i+1$
			\State {\it Append} to $RemovedRows[i]$ all $AvailableRows[s]$  for which there is a vertex $u$
			\NoNumber{\hspace{1.2 cm}such that the state of rows}
			$AvailableRows[s]$ and $AvailableRows[j]$ both assign 1 to $u$
			\State {\it Remove} all elements of $RemovedRows[i]$ from $AvailableRows$
			\EndIf
			
			\EndWhile
		\end{algorithmic}
	\end{algorithm}
	
	If the output of Algorithm \ref{algorithm1} has less than $n$ elements, then $H$ has no admissible $n$-coloring---because else Theorem \ref{separable} guarantees that there are $n$ two-valued states partitioning the logic, in which case Algorithm \ref{algorithm1} would have given $\vert A \vert = n$. If $\vert A \vert =n$, then $A$ is a list of $n$ rows in $T(H)$, each of which corresponds to a color class of an $n$-coloring of $H$. In other words, when $s\in A$, the two-valued state $t_s$ presents the color class consisting all the vertices it assigns, or maps to, 1. It is evident that the resulting color classes are independent sets while Formula \ref{color-all-two-states} implies that they cover all the vertices. Consequently, $A$ induces a proper $n$-coloring on vertices of $H$.
	
	Algorithm \ref{algorithm1} is not highly efficient in finding an $n$-coloring for $H$. The main reason is that in the worst case study it has to check all the two-valued states of $H$ whose number, i.e. $nTS(H)$, can grow exponentially in terms of the clique number and number of vertices and hyperedges.
	
	Moreover, one could conjecture that Algorithm~\ref{algorithm1} could be modified to render a coloring even if the (hyper)graph is not $n$-partitionable, in which case Theorem~\ref{separable} does not apply.
	Because even if one has exhausted all combinations of two-valued states one could still attempt to ``complete'' the coloring by identifying the missing colors with
	``suitable segments'' of the remaining two-valued states (if there are any leftovers). Of course, in this way, the column sums of all the respective two valued states cannot be 1, and hence Formula~(\ref{color-all-two-states}) is no longer valid.
	In any case, Brooks' theorem~\cite{Brooks1941,Lovasz1975}---stating that for any connected undirected graph $G$, the chromatic number of $G$ is at most its maximum degree (the maximal number of edges that are incident to some vertex) $\Delta$  unless $G$ is a complete graph or an odd cycle, in which case the chromatic number is $\Delta + 1$---and its generalization to hypergraphs~\cite[page 45, Theorem 3.2]{Bretto-MR3077516} yield an upper bound for the chromatic number of such (hyper)graphs.
	
	
	\section{Summary and concluding remarks}

	We have presented a constructive, algorithmic way to generate a coloring of a (hyper)graph from its set of two-valued states.  The only criterion for the success of this approach is the assertion that the respective hypergraph is semi-perfect, that is, its chromatic number equals the clique number of its 2-section.
	We have been able to find a ``compact'' partition logic within the logical states of the hypergraph by showing that $n$-colorability  is equivalent to finding a partition logic based on exactly $n$ two-valued states. We also presented a detailed algorithm for constructively finding this partition logic and its associated coloring.
	
	With regard to representing and reconstructing (hyper)graphs or logics in terms of their two-valued states, in particular, regarding separability of vertices or elementary propositions, we conjecture that there exist quantum logics with a separable set of two-valued states that cannot be reconstructed from these states. We have presented a hypergraph, namely  $B(G)$ depicted in Figure~\ref{TIFS-non-Rec} of Section~\ref{Rec-B(H)} with a TIFS gadget such as the one depicted in Figure~\ref{Baba-Taher}, that has this characteristic but we could not find a faithful orthogonal representation in a Hilbert space.
	
	Yet, stronger forms of separability, in particular, perfect separability, can be identified that allow (hyper)graphs or logics to be represented and reconstructed in terms of their two-valued states (that is, by their Travis matrices). In addition, while the conditions on perfectly separable (hyper)graphs are rather strong, one can be certain that such a reconstruction exists.
	
	Indeed, such a reconstruction helps to directly identify mutually perpendicular elementary propositions, and thus the contexts corresponding to the maximal operators they form: if an orthogonality (hyper)graph is reconstructible from its set of two-valued states we can deduce the mutual orthogonality of the elementary quantum propositions by just looking at these two-valued states. This facilitates the construction of the (mutually perpendicular) orthogonal operators in the spectral sums associated with the contexts, and thus supports finding a global faithful orthogonal representation, i.~e., the assignment of vectors to vertices, of (hyper)graphs.
	
	Stated differently, we showed that there is a class of hypergraphs, namely perfectly separable ones, that are always reconstructable from their two-valued states.
	However, not all separable graphs are guaranteed to be reconstructible by these means.
	
	Hence, while for perfectly separable (hyper)graphs we can be certain that they can be reconstructed; and for Kochen-Specker type (hyper)graphs that they cannot be reconstructed because there is no two-valued state associated with any classical value assignment, for the remaining (hyper)graphs reconstructability remains an open question.

	\appendix

	\section{Examples}
\label{AppendixA}
	\subsection{Triangle logic}
	
	The coloring procedure of the triangle hypergraph is depicted in Figure~\ref{2020-f-chroma-triangle3}.
	Consider the set of all four two-valued states on the six atoms which can be tabulated by a
	(compactified) Travis matrix $T_{ij}$
	whose rows indicate the
	$i$th state $s_i$ and whose columns
	indicate the atoms $a_j$, respectively; that is, $T_{ij}=s_i(a_j)$:
	\begin{equation}
		T_{ij}=\begin{pmatrix}
			{\color{red}1}&{\color{red}0}&{\color{red}0}&{\color{red}1}&{\color{red}0}&{\color{red}0}\\
			{\color{blue}0}&{\color{blue}0}&{\color{blue}1}&{\color{blue}0}&{\color{blue}0}&{\color{blue}1}\\
			{\color{green}0}&{\color{green}1}&{\color{green}0}&{\color{green}0}&{\color{green}1}&{\color{green}0}\\
			0&1&0&1&0&1
		\end{pmatrix}
		.
	\end{equation}
	It is not too difficult to see that the first three measures, represented by the first three row vectors of the
	Travis matrix, add up to
	$
	\begin{pmatrix}
		1,1,1,1,1,1
	\end{pmatrix}
	$. They can thus be taken as the basis of a coloring.
	
	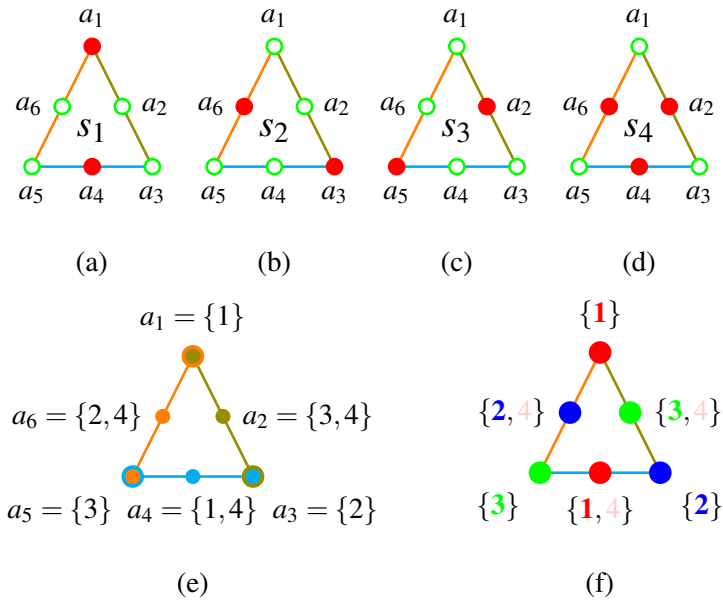
\begin{figure}
		\begin{center}
			\begin{tabular}{ c c c c }
				\begin{tikzpicture}  [scale=0.8]
					
					\tikzstyle{every path}=[line width=1pt]
					
					\newdimen\ms
					\ms=0.1cm
					\tikzstyle{s1}=[fill=red,draw=red,circle,inner sep=2]
					\tikzstyle{c1}=[fill=white,draw=green,circle,inner sep={\ms/8},minimum size=2*\ms]
					

					\coordinate (a1) at  (1,2);
					\coordinate (a2) at (1.5,1);
					\coordinate (a3) at (2,0);
					\coordinate (a4) at (1,0);
					\coordinate (a5) at (0,0);
					\coordinate (a6) at (0.5,1);
					\coordinate (c) at (1,0.6);
					
					
					\draw [color=olive] (a1) -- (a3);
					\draw [color=cyan] (a3) -- (a5);
					\draw [color=orange] (a5) -- (a1);

					
					\draw (a1) coordinate[s1,label=above:$a_1$];
					\draw (a2) coordinate[c1,label=right:$a_2$];
					\draw (a3) coordinate[c1,label=below:$a_3$];
					\draw (a4) coordinate[s1,label=below:$a_4$];
					\draw (a5) coordinate[c1,label=below:$a_5$];
					\draw (a6) coordinate[c1,label=left:$a_6$];
					\node at (c) {\large $s_1$};
					
				\end{tikzpicture}
				&
				\begin{tikzpicture}  [scale=0.8]
					
					\tikzstyle{every path}=[line width=1pt]
					
					\newdimen\ms
					\ms=0.1cm
					\tikzstyle{s1}=[fill=red,draw=red,circle,inner sep=2]
					\tikzstyle{c1}=[fill=white,draw=green,circle,inner sep={\ms/8},minimum size=2*\ms]
					

					\coordinate (a1) at  (1,2);
					\coordinate (a2) at (1.5,1);
					\coordinate (a3) at (2,0);
					\coordinate (a4) at (1,0);
					\coordinate (a5) at (0,0);
					\coordinate (a6) at (0.5,1);
					\coordinate (c) at (1,0.6);
					
					
					\draw [color=olive] (a1) -- (a3);
					\draw [color=cyan] (a3) -- (a5);
					\draw [color=orange] (a5) -- (a1);

					
					\draw (a1) coordinate[c1,label=above:$a_1$];
					\draw (a2) coordinate[c1,label=right:$a_2$];
					\draw (a3) coordinate[s1,label=below:$a_3$];
					\draw (a4) coordinate[c1,label=below:$a_4$];
					\draw (a5) coordinate[c1,label=below:$a_5$];
					\draw (a6) coordinate[s1,label=left:$a_6$];
					\coordinate (c) at (1,0.6);
					\node at (c) {\large $s_2$};
					
				\end{tikzpicture}
				&
				\begin{tikzpicture}  [scale=0.8]
					
					\tikzstyle{every path}=[line width=1pt]
					
					\newdimen\ms
					\ms=0.1cm
					\tikzstyle{s1}=[fill=red,draw=red,circle,inner sep=2]
					\tikzstyle{c1}=[fill=white,draw=green,circle,inner sep={\ms/8},minimum size=2*\ms]
					

					\coordinate (a1) at  (1,2);
					\coordinate (a2) at (1.5,1);
					\coordinate (a3) at (2,0);
					\coordinate (a4) at (1,0);
					\coordinate (a5) at (0,0);
					\coordinate (a6) at (0.5,1);
					\coordinate (c) at (1,0.6);
					
					
					\draw [color=olive] (a1) -- (a3);
					\draw [color=cyan] (a3) -- (a5);
					\draw [color=orange] (a5) -- (a1);

					
					\draw (a1) coordinate[c1,label=above:$a_1$];
					\draw (a2) coordinate[s1,label=right:$a_2$];
					\draw (a3) coordinate[c1,label=below:$a_3$];
					\draw (a4) coordinate[c1,label=below:$a_4$];
					\draw (a5) coordinate[s1,label=below:$a_5$];
					\draw (a6) coordinate[c1,label=left:$a_6$];
					\node at (c) {\large $s_3$};
					
				\end{tikzpicture}
				&
				\begin{tikzpicture}  [scale=0.8]
					
					\tikzstyle{every path}=[line width=1pt]
					
					\newdimen\ms
					\ms=0.1cm
					\tikzstyle{s1}=[fill=red,draw=red,circle,inner sep=2]
					\tikzstyle{c1}=[fill=white,draw=green,circle,inner sep={\ms/8},minimum size=2*\ms]
					

					\coordinate (a1) at  (1,2);
					\coordinate (a2) at (1.5,1);
					\coordinate (a3) at (2,0);
					\coordinate (a4) at (1,0);
					\coordinate (a5) at (0,0);
					\coordinate (a6) at (0.5,1);
					\coordinate (c) at (1,0.6);
					
					
					\draw [color=olive] (a1) -- (a3);
					\draw [color=cyan] (a3) -- (a5);
					\draw [color=orange] (a5) -- (a1);

					
					\draw (a1) coordinate[c1,label=above:$a_1$];
					\draw (a2) coordinate[s1,label=right:$a_2$];
					\draw (a3) coordinate[c1,label=below:$a_3$];
					\draw (a4) coordinate[s1,label=below:$a_4$];
					\draw (a5) coordinate[c1,label=below:$a_5$];
					\draw (a6) coordinate[s1,label=left:$a_6$];
					\node at (c) {\large $s_4$};
					
				\end{tikzpicture}
				\\
				(a)&(b)&(c)&(d)\\
			\end{tabular}
			\\
			\begin{tabular}{ c c c}
				\begin{tikzpicture}  [scale=0.8]
					
					\tikzstyle{every path}=[line width=1pt]
					
					\newdimen\ms
					\ms=0.1cm
					\tikzstyle{c2}=[circle,inner sep={\ms/8},minimum size=3*\ms]
					\tikzstyle{c1}=[circle,inner sep={\ms/8},minimum size=2*\ms]
					

					\coordinate (a1) at  (1,2);
					\coordinate (a2) at (1.5,1);
					\coordinate (a3) at (2,0);
					\coordinate (a4) at (1,0);
					\coordinate (a5) at (0,0);
					\coordinate (a6) at (0.5,1);
					
					
					\draw [color=olive] (a1) -- (a3);
					\draw [color=cyan] (a3) -- (a5);
					\draw [color=orange] (a5) -- (a1);

					
					\draw (a1) coordinate[c2,fill=orange,label=above:${a_1=\{1\}}$];
					\draw (a1) coordinate[c1,fill=olive];
					
					\draw (a2) coordinate[c1,fill=olive,label=right:${a_2=\{3,4\}}$];
					
					\draw (a3) coordinate[c2,fill=olive,label=below right:${a_3=\{2\}}$];
					\draw (a3) coordinate[c1,fill=cyan];
					
					\draw (a4) coordinate[c1,fill=cyan,label=below:${a_4=\{1,4\}}$];
					
					\draw (a5) coordinate[c2,fill=cyan,label=below left:${a_5=\{3\}}$];
					\draw (a5) coordinate[c1,fill=orange];
					
					\draw (a6) coordinate[c1,fill=orange,label=left:${a_6=\{2,4\}}$];
					
				\end{tikzpicture}
				&
				$\qquad$
				&
				\begin{tikzpicture}  [scale=0.8]
					
					\tikzstyle{every path}=[line width=1pt]
					
					\newdimen\ms
					\ms=0.1cm
					\tikzstyle{c2}=[circle,inner sep={\ms/8},minimum size=3*\ms]
					\tikzstyle{c1}=[circle,inner sep={\ms/8},minimum size=2*\ms]
					

					\coordinate (a1) at  (1,2);
					\coordinate (a2) at (1.5,1);
					\coordinate (a3) at (2,0);
					\coordinate (a4) at (1,0);
					\coordinate (a5) at (0,0);
					\coordinate (a6) at (0.5,1);
					
					
					\draw [color=olive] (a1) -- (a3);
					\draw [color=cyan] (a3) -- (a5);
					\draw [color=orange] (a5) -- (a1);

					
					\draw (a1) coordinate[c2,fill=red,label=above:${\{{\color{red} \bf 1}\}}$];
					
					\draw (a2) coordinate[c2,fill=green,label=right:${\{{\color{green} \bf 3},{\color{red!20!white} 4}\}}$];
					
					\draw (a3) coordinate[c2,fill=blue,label=below right:${\{{\color{blue} \bf 2}\}}$];

					\draw (a4) coordinate[c2,fill=red,label=below:${\{{\color{red} \bf 1},{\color{red!20!white} 4}\}}$];
					
					\draw (a5) coordinate[c2,fill=green,label=below left:${\{{\color{green} \bf 3}\}}$];

					\draw (a6) coordinate[c2,fill=blue,label=left:${\{{\color{blue} \bf 2},{\color{red!20!white} 4}\}}$];
					
				\end{tikzpicture}
				\\
				(e)&&(f)
			\end{tabular}
		\end{center}
		\caption{\label{2020-f-chroma-triangle3}
			One (nonunique) coloring~(f) construction of
			the triangle hypergraph of the logic: first compose a (nonunique)
			canonical partition logic~(e) from enumerating the set of all 4 two-valued states depicted in (a)--(d).
			Then choose the context $\{a_1,a_2,a_3\}$, and from this context choose the atom $a_1=\{1\}$.
			Now identify the first color (red) with the index 1, thereby identifying $a_1=\{1\}$ and $a_4=\{1,4\}$ with red.
			Then, delete the index number $4$ from every atom; that is, $a_2=\{3,4\}\rightarrow \{3\}$ and $a_6=\{2,4\}\rightarrow \{2\}$.
			Finally, identify 3 with the second color (green) and 2 with the third color (blue),
			thereby identifying $a_2$ and $a_5$ with green, and $a_3$ and $a_6$ with blue, respectively.
			Note that $s_1$, $s_2$, and $s_3$ ``generate'' a 3-partitioning of the set of atoms $\{a_1,\ldots ,a_6\}$ of this logic.
		}
	\end{figure}
	
	\subsection{House or pentagon or pentagram logic}
	
	The Travis matrix of the house or pentagon or pentagram logic depicted in Fig.~7
	is a matrix representation of its 11 dispersion free states~\cite{wright:pent}
	\begin{equation}
		T_{ij}=\begin{pmatrix}
			{\color{blue}1}& {\color{blue}0}& {\color{blue}0}& {\color{blue}1}& {\color{blue}0}& {\color{blue}1}& {\color{blue}0}& {\color{blue}1}& {\color{blue}0}& {\color{blue}0}  \\
			1& 0& 0& 1& 0& 0& 1& 0& 0& 0  \\
			1& 0& 0& 0& 1& 0& 0& 1& 0& 0  \\
			0& 1& 0& 1& 0& 1& 0& 1& 0& 1  \\
			0& 1& 0& 1& 0& 1& 0& 0& 1& 0  \\
			0& 1& 0& 1& 0& 0& 1& 0& 0& 1  \\
			0& 1& 0& 0& 1& 0& 0& 1& 0& 1  \\
			{\color{green}0}& {\color{green}1}& {\color{green}0}& {\color{green}0}& {\color{green}1}& {\color{green}0}& {\color{green}0}& {\color{green}0}& {\color{green}1}& {\color{green}0}  \\
			0& 0& 1& 0& 0& 1& 0& 1& 0& 1  \\
			0& 0& 1& 0& 0& 1& 0& 0& 1& 0  \\
			{\color{red}0}& {\color{red}0}& {\color{red}1}& {\color{red}0}& {\color{red}0}& {\color{red}0}& {\color{red}1}& {\color{red}0}& {\color{red}0}& {\color{red}1}
		\end{pmatrix}
		.
	\end{equation}
	A coloring can be obtained from the earlier mentioned construction
	which results in three states partitioning all 10 atoms.
	The associated 1st, the 8th and the 11th row vectors
	of $T_{ij}$  are partitioning the 10 atoms.

	\begin{figure}
		\begin{center}
			\begin{tikzpicture}  [scale=1]
				
				\tikzstyle{every path}=[line width=1pt]
				
				\newdimen\ms
				\ms=0.1cm
				\tikzstyle{s1}=[color=red,rectangle,inner sep=3.5]
				\tikzstyle{c3}=[circle,inner sep={\ms/8},minimum size=5*\ms]
				\tikzstyle{c2}=[circle,inner sep={\ms/8},minimum size=3*\ms]
				\tikzstyle{c1}=[circle,inner sep={\ms/8},minimum size=2*\ms]
				
				
				\coordinate (a1) at (0,2);
				\coordinate (a2) at (0,1);
				\coordinate (a3) at (0,0);
				\coordinate (a4) at (1,0);
				\coordinate (a5) at (2,0);
				\coordinate (a6) at (2,1);
				\coordinate (a7) at (2,2);
				\coordinate (a8) at (1.5,{2+(3.5-2)/2});
				\coordinate (a9) at (1,3.5);
				\coordinate (a10) at (0.5,{2+(3.5-2)/2});
				
				
				\draw [color=orange] (a1) -- (a3);
				\draw [color=blue] (a3) -- (a5);
				\draw [color=red] (a5) -- (a7);
				\draw [color=green] (a7) -- (a9);
				\draw [color=gray] (a9) -- (a1);
				
				
				\draw (a1) coordinate[c2,fill=blue,label=left:{$\{{\color{blue}1},2,3\}$}];
				
				\draw (a2) coordinate[c2,fill=green,label=left:{$\{ 4,5,6,7,{\color{green}8} \}$}];
				
				\draw (a3) coordinate[c2,fill=red,label=below left:{$\{  9,10,{\color{red}11}\}$}];
				
				\draw (a4) coordinate[c2,fill=blue,label=below:{$\{ {\color{blue}1},2,4,5,6 \}$}];
				
				\draw (a5) coordinate[c2,fill=green,label=below right:{$\{ 3,7,{\color{green}8} \}$}];
				
				\draw (a6) coordinate[c2,fill=blue,label=right:{$\{ {\color{blue}1},4,5,9,10 \}$}];
				
				\draw (a7) coordinate[c2,fill=red,label=right:{$\{ 2,6,{\color{red}11} \}$}];
				
				\draw (a8) coordinate[c2,fill=blue,label=above right:{$\{ {\color{blue}1},3,4,7,9 \}$}];
				
				\draw (a9) coordinate[c2,fill=green,label=above:{$\{ 5,{\color{green}8},10 \}$}];
				
				\draw (a10) coordinate[c2,fill=red,label=above left:{$\{ 4,6,7,9,{\color{red}11} \}$}];
				
			\end{tikzpicture}
		\end{center}
		\caption{\label{2020-f-chroma-pentagon3}
			Coloring scheme of the house or pentagon or pentagram logic from the set of two-valued states.
		}
	\end{figure}
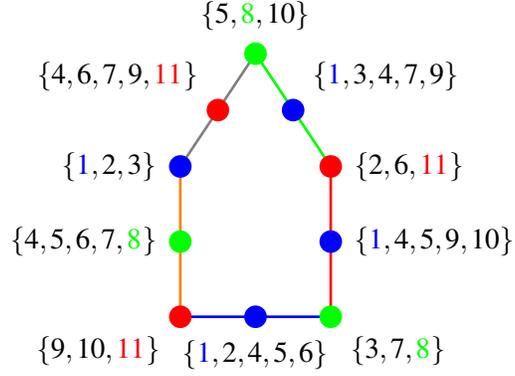

	\subsection{Specker bug gadget}
	\label{Specker-Bug}
	
	The hypergraph depicted in Figure~\ref{2020-f-SpeckerBug}
	is a minimal~\cite{2018-minimalYIYS} true-implies false gadget introduced by
	Kochen and Specker~\cite[Fig.~1, p.~182]{kochen2} (reprinted in~\cite{specker-ges}, see also ~\cite[Figure~1, p.~123]{Greechie1974}, among others).
	It is a subgraph of $G_{32}$ introduced later in Figure~\ref{2020-f-GreechieG32}.
	Its Travis matrix is
	\begin{equation}\label{Specker-Bug-Travice}
		T_{ij}=\begin{pmatrix}
			{\color{blue}1}& {\color{blue}0}& {\color{blue}0}& {\color{blue}1}& {\color{blue}0} &{\color{blue}1}& {\color{blue}0}& {\color{blue}0}& {\color{blue}1}& {\color{blue}0}& {\color{blue}0}& {\color{blue}0}& {\color{blue}0} \\
			1& 0& 0& 0& 1 &0& 0& 1& 0& 1& 0& 0& 0 \\
			1& 0& 0& 0& 1 &0& 0& 0& 1& 0& 0& 0& 1 \\
			0& 1& 0& 1& 0 &1& 0& 1& 0& 0& 1& 0& 0 \\
			0& 1& 0& 1& 0 &1& 0& 0& 1& 0& 0& 1& 0 \\
			0& 1& 0& 1& 0 &0& 1& 0& 0& 0& 1& 0& 0 \\
			{\color{green}0}& {\color{green}1}& {\color{green}0}& {\color{green}0}& {\color{green}1} &{\color{green}0}& {\color{green}0}& {\color{green}1}& {\color{green}0}& {\color{green}1}& {\color{green}0}& {\color{green}1}& {\color{green}0} \\
			0& 1& 0& 0& 1 &0& 0& 1& 0& 0& 1& 0& 1 \\
			0& 1& 0& 0& 1 &0& 0& 0& 1& 0& 0& 1& 1 \\
			0& 0& 1& 0& 0 &1& 0& 1& 0& 1& 0& 1& 0 \\
			0& 0& 1& 0& 0 &1& 0& 1& 0& 0& 1& 0& 1 \\
			0& 0& 1& 0& 0 &1& 0& 0& 1& 0& 0& 1& 1 \\
			0& 0& 1& 0& 0 &0& 1& 0& 0& 1& 0& 1& 0 \\
			{\color{red}0}& {\color{red}0}& {\color{red}1}& {\color{red}0}& {\color{red}0} &{\color{red}0}& {\color{red}1}& {\color{red}0}& {\color{red}0}& {\color{red}0}& {\color{red}1}& {\color{red}0}& {\color{red}1}
		\end{pmatrix}
		.
	\end{equation}
	
	\begin{figure}
		\begin{center}
			\begin{tikzpicture}  [scale=0.8]
				
				\newdimen\ms
				\ms=0.05cm
				
				\tikzstyle{every path}=[line width=1pt]
				
				\tikzstyle{c3}=[circle,inner sep={\ms/8},minimum size=6*\ms]
				\tikzstyle{c2}=[circle,inner sep={\ms/8},minimum size=4*\ms]
				\tikzstyle{c1}=[circle,inner sep={\ms/8},minimum size=0.8*\ms]
				
				\newdimen\R
				\R=30mm     
				
				
				
				\path
				({ 180 - 0 * 360 /6}:\R      ) coordinate(1)
				({ 180 - 30 - 0 * 360 /6}:{\R * sqrt(3)/2}      ) coordinate(2)
				({ 180 - 1 * 360 /6}:\R   ) coordinate(3)
				({ 180 - 30 - 1 * 360 /6}:{\R * sqrt(3)/2}   ) coordinate(4)
				({ 180 - 2 * 360 /6}:\R  ) coordinate(5)
				({ 180 - 30 - 2 * 360 /6}:{\R * sqrt(3)/2}  ) coordinate(6)
				({ 180 - 3 * 360 /6}:\R  ) coordinate(7)
				({ 180 - 30 - 3 * 360 /6}:{\R * sqrt(3)/2}  ) coordinate(8)
				({ 180 - 4 * 360 /6}:\R     ) coordinate(9)
				({ 180 - 30 - 4 * 360 /6}:{\R * sqrt(3)/2}     ) coordinate(10)
				({ 180 - 5 * 360 /6}:\R     ) coordinate(11)
				({ 180 - 30 - 5 * 360 /6}:{\R * sqrt(3)/2}     ) coordinate(12)
				;
				
				
				\draw [color=cyan] (1) -- (2) -- (3);
				\draw [color=red] (3) -- (4) -- (5);
				\draw [color=green] (5) -- (6) -- (7);
				\draw [color=blue] (7) -- (8) -- (9);
				\draw [color=magenta] (9) -- (10) -- (11);    %
				\draw [color=olive] (11) -- (12) -- (1);    %
				\draw [color=teal] (4) -- (10)  coordinate[pos=0.5]  (13);
				
				%
				%
				\draw (1) coordinate[c3,fill=blue,label={left: $\{ {\color{blue}1},2,3\} $}];   %
				\draw (2) coordinate[c3,fill=green,label={above left: $\{ 4,5,6,{\color{green}7},8,9 \}$}];    %
				\draw (3) coordinate[c3,fill=red,label={above left: $\{10,11,12,13,{\color{red}14} \} $}]; %
				\draw (4) coordinate[c3,fill=blue,label={above: $\{ {\color{blue}1},4,5,6\}$}];  %
				\draw (5) coordinate[c3,fill=green,label={above right: $\{ 2,3,{\color{green}7},8,9\} $}];  %
				\draw (6) coordinate[c3,fill=blue,label={above right: $\{ {\color{blue}1},4,5,10,11,12\} $}];
				\draw (7) coordinate[c3,fill=red,label={right: $\{ 6,13,{\color{red}14}\}$}];  %
				\draw (8) coordinate[c3,fill=green,label={below right: $\{ 2,4,{\color{green}7},8,10,11\}$}];  %
				\draw (9) coordinate[c3,fill=blue,label={below right: $\{ {\color{blue}1},3,5,9,12\}$}];
				\draw (10) coordinate[c3,fill=green,label={below: $\{ 2,{\color{green}7},10,13\}$}];  %
				\draw (11) coordinate[c3,fill=red,label={below left: $\{ 4,6,8,11,{\color{red}14}\}$}];  %
				\draw (12) coordinate[c3,fill=green,label={below left: $\{ 5,{\color{green}7},9,10,12,13\}$}];
				\draw (13) coordinate[c3,fill=red,label={right: $\{3,8,9,$}];  %
				\draw (13) coordinate[c3,fill=red,label={below right: $11,12,{\color{red}14}\}$}];  %
			\end{tikzpicture}
		\end{center}
		\caption{\label{2020-f-SpeckerBug}
			Coloring scheme of the ``Specker bug'' gadget~\cite{kochen2,Greechie1974} from two-valued states.
			The set-theoretic representation is in terms of
			the canonical partition logic as an equipartitioning of the set $\{1,2,\ldots,14\}$
			obtained from all 14 two-valued states on this gadget.
		}
	\end{figure}
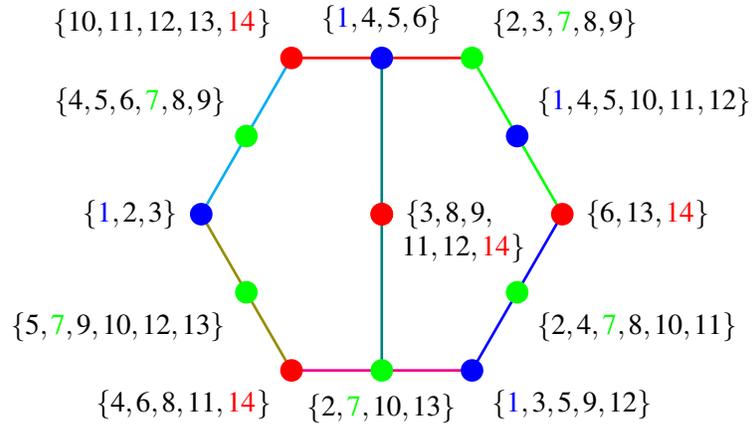
	
	\subsection{The underlying hypergraph of $B(G)$}\label{sec:underlying}
	
	All the two-valued states of the structure introduced in Figure~\ref{TIFS-non-Rec} have to assign the vertices $a$, $b$, $c$, $a'$, $b'$, $c'$ and $a''$, $b''$, $c''$ the same combination of two-valued states as for the hypergraph of Figure~\ref{Fig:underlying}.
	
	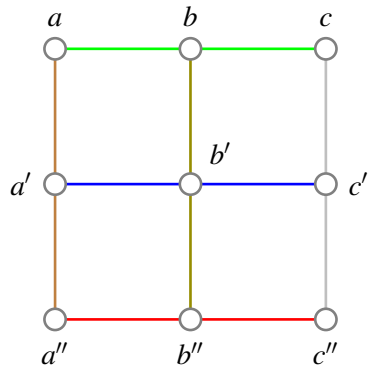
\begin{figure}
		\begin{center}
			\begin{tikzpicture}  [scale=0.9]
				
				\tikzstyle{every path}=[line width=1pt]
				
				\newdimen\ms
				\ms=0.1cm
				\tikzstyle{s1}=[color=red,rectangle,inner sep=3.5]
				\tikzstyle{c3}=[circle,inner sep={\ms/8},minimum size=4*\ms]
				\tikzstyle{c2}=[circle,inner sep={\ms/8},minimum size=3*\ms]
				\tikzstyle{c1}=[draw=gray,fill=white,circle,inner sep={\ms/1}]
				\tikzstyle{cs1}=[circle,inner sep={\ms/8},minimum size=1*\ms]

				
				\coordinate (a) at ( 0,4);
				\coordinate (b) at ( 2,4);
				\coordinate (c) at ( 4,4);

				\coordinate (a1) at ( 0,2);
				\coordinate (b1) at ( 2,2);
				\coordinate (c1) at ( 4,2);

				\coordinate (a2) at ( 0,0);
				\coordinate (b2) at ( 2,0);
				\coordinate (c2) at ( 4,0);

				
				\draw [color=green] (a) -- (c);
				\draw [color=blue] (a1) -- (c1);
				\draw [color=red] (a2) -- (c2);

				\draw [color=brown] (a) -- (a2);
				\draw [color=olive] (b) -- (b2);
				\draw [color=lightgray] (c) -- (c2);


				
				\draw (a)  coordinate[c1,label=above:{$a$}];
				\draw (b)  coordinate[c1,label=above:{$b$}];
				\draw (c)  coordinate[c1,label=above:{$c$}];
				\draw (a1) coordinate[c1,label=left:{$a'$}];
				\draw (b1) coordinate[c1,label=above right:{$b'$}];
				\draw (c1) coordinate[c1,label=right:{$c'$}];
				\draw (a2) coordinate[c1,label=below:{$a''$}];
				\draw (b2) coordinate[c1,label=below:{$b''$}];
				\draw (c2) coordinate[c1,label=below:{$c''$}];

			\end{tikzpicture}
		\end{center}
		\caption{\label{Fig:underlying}
			From the point of view of two-valued states, for which all the nine $G_i$'s in Figure~\ref{TIFS-non-Rec}, $1\le i\le 9$  are TIFS, the three ``new'' contexts
			$\{a, b, c\}$, $\{a', b', c'\}$ and $\{a'', b'',c''\}$ are formed
			through three triples of TIFS
			$\{G_1,G_2,G_3\}$, $\{G_4,G_5,G_6\}$, and $\{G_7,G_8,G_9\}$, respectively; thereby rendering a tightly bi-connected hypergraph underlying the one depicted in Figure~\ref{TIFS-non-Rec}. Note that the vertices of each row in the original graph of $B(G)$ do not lie on a context, but here in the underlying hypergraph they are.
		}
	\end{figure}
	
	The Travis matrix of this tightly bi-connected hypergraph is as follows (columns from left to right correspond to $a$, $b$, $c$, $a'$, $b'$, $c'$, $a''$, $b''$ and $c''$):
	\begin{equation}\label{underlying-Travice}
		T_{ij}=\begin{pmatrix}
			1 &  0 & 0 & 0 & 1 & 0 & 0 & 0 & 1    \\
			1 &  0 & 0 & 0 & 0 & 1 & 0 & 1 & 0    \\
			0 &  1 & 0 & 1 & 0 & 0 & 0 & 0 & 1    \\
			0 &  1 & 0 & 0 & 0 & 1 & 1 & 0 & 0    \\
			0 &  0 & 1 & 1 & 0 & 0 & 0 & 1 & 0    \\
			0 &  0 & 1 & 0 & 1 & 0 & 1 & 0 & 0    \\
		\end{pmatrix}
		.
	\end{equation}
	
	\subsection{``Tight GHZ'' logic}
	
	The hypergraph depicted in Figure~\ref{2020-f-GHZ}
	is a sublogic  of the observables in the Greenberger-Horn-Zeilinger setup~\cite{svozil-2020-ghz}.
	Its Travis matrix is
	\begin{equation}\label{GHZ-tight-Travice}
		T_{ij}=\begin{pmatrix}
			{\color{blue}1 } & {\color{blue}  0 } & {\color{blue} 0 } & {\color{blue} 0 } & {\color{blue} 0 } & {\color{blue} 0 } & {\color{blue} 1 } & {\color{blue} 0 } & {\color{blue} 0 } & {\color{blue} 0 } & {\color{blue} 0 } & {\color{blue} 1 } & {\color{blue} 0 } & {\color{blue} 1 } & {\color{blue} 0 } & {\color{blue} 0 }   \\
			1 &  0 & 0 & 0 & 0 & 0 & 0 & 1 & 0 & 1 & 0 & 0 & 0 & 0 & 1 & 0    \\
			0 &  1 & 0 & 0 & 0 & 0 & 1 & 0 & 1 & 0 & 0 & 0 & 0 & 0 & 0 & 1    \\
			{\color{green}0 } & {\color{green}  1 } & {\color{green} 0 } & {\color{green} 0 } & {\color{green} 0 } & {\color{green} 0 } & {\color{green} 0 } & {\color{green} 1 } & {\color{green} 0 } & {\color{green} 0 } & {\color{green} 1 } & {\color{green} 0 } & {\color{green} 1 } & {\color{green} 0 } & {\color{green} 0 } & {\color{green} 0 }   \\
			{\color{red}0 } & {\color{red}  0 } & {\color{red} 1 } & {\color{red} 0 } & {\color{red} 1 } & {\color{red} 0 } & {\color{red} 0 } & {\color{red} 0 } & {\color{red} 0 } & {\color{red} 1 } & {\color{red} 0 } & {\color{red} 0 } & {\color{red} 0 } & {\color{red} 0 } & {\color{red} 0 } & {\color{red} 1   } \\
			0 &  0 & 1 & 0 & 0 & 1 & 0 & 0 & 0 & 0 & 0 & 1 & 1 & 0 & 0 & 0    \\
			0 &  0 & 0 & 1 & 1 & 0 & 0 & 0 & 0 & 0 & 1 & 0 & 0 & 1 & 0 & 0    \\
			{\color{cyan}0 } & {\color{cyan}  0 } & {\color{cyan} 0 } & {\color{cyan} 1 } & {\color{cyan} 0 } & {\color{cyan} 1 } & {\color{cyan} 0 } & {\color{cyan} 0 } & {\color{cyan} 1 } & {\color{cyan} 0 } & {\color{cyan} 0 } & {\color{cyan} 0 } & {\color{cyan} 0 } & {\color{cyan} 0 } & {\color{cyan} 1 } & {\color{cyan} 0 }
		\end{pmatrix}
		.
	\end{equation}
	
	The coloring is depicted in Figure~\ref{2020-f-GHZ}.
	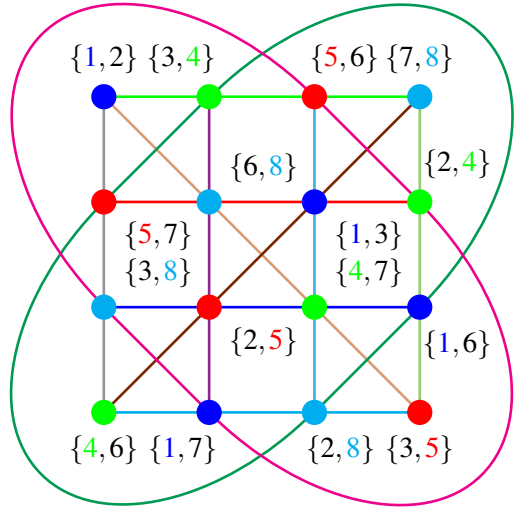
\begin{figure}
		\begin{center}
			\begin{tikzpicture}  [scale=0.7]
				
				\tikzstyle{every path}=[line width=1pt]
				
				\newdimen\ms
				\ms=0.1cm
				\tikzstyle{s1}=[color=red,rectangle,inner sep=3.5]
				\tikzstyle{c3}=[circle,inner sep={\ms/8},minimum size=4*\ms]
				\tikzstyle{c2}=[circle,inner sep={\ms/8},minimum size=3*\ms]
				\tikzstyle{c1}=[circle,inner sep={\ms/8},minimum size=2*\ms]
				\tikzstyle{cs1}=[circle,inner sep={\ms/8},minimum size=1*\ms]

				
				\coordinate (uuu) at ( 0,6);
				\coordinate (uuv) at ( 2,6);
				\coordinate (uvu) at ( 4,6);
				\coordinate (uvv) at ( 6,6);
				\coordinate (vuu) at ( 8,6);
				\coordinate (vuv) at (10,6);
				\coordinate (vvu) at (12,6);
				\coordinate (vvv) at (14,6);

				\coordinate (ucc) at ( 0,4);
				\coordinate (vcc) at ( 2,4);
				\coordinate (ucd) at ( 4,4);
				\coordinate (vcd) at ( 6,4);
				\coordinate (udc) at ( 8,4);
				\coordinate (vdc) at (10,4);
				\coordinate (udd) at (12,4);
				\coordinate (vdd) at (14,4);
				
				\coordinate (cuc) at ( 0,2);
				\coordinate (cvc) at ( 2,2);
				\coordinate (cud) at ( 4,2);
				\coordinate (cvd) at ( 6,2);
				\coordinate (duc) at ( 8,2);
				\coordinate (dvc) at (10,2);
				\coordinate (dud) at (12,2);
				\coordinate (dvd) at (14,2);
				
				\coordinate (ccu) at ( 0,0);
				\coordinate (ccv) at ( 2,0);
				\coordinate (cdu) at ( 4,0);
				\coordinate (cdv) at ( 6,0);
				\coordinate (dcu) at ( 8,0);
				\coordinate (dcv) at (10,0);
				\coordinate (ddu) at (12,0);
				\coordinate (ddv) at (14,0);

				
				\draw [color=cyan] (ccu) -- (cdv);
				\draw [color=blue] (cuc) -- (cvd);
				\draw [color=red] (ucc) -- (vcd);
				\draw [color=green] (uuu) -- (uvv);
				
				\draw [color=Gray] (uuu) -- (ccu);
				\draw [color=Plum] (uuv) -- (ccv);
				\draw [color=CornflowerBlue] (uvu) -- (cdu);
				\draw [color=YellowGreen] (uvv) -- (cdv);

				\draw [color=Tan] (uuu) -- (cdv);
				\draw [color=Brown] (uvv) -- (ccu);

				\draw [color=ForestGreen] (uuv) -- (ucc);
				\draw [color=ForestGreen](cdu) -- (cvd);
				\draw [rotate=225,color=ForestGreen] (cvd) arc (90:270:4.5 and 2.82);
				\draw[rotate=45,color=ForestGreen] (ucc) arc (90:270:4.5 and 2.82);
				
				\draw [color=Magenta] (cuc) -- (ccv);
				\draw [color=Magenta] (uvu) -- (vcd);
				\draw[rotate=315,color=Magenta] (uvu) arc (90:270:4.5 and 2.82);
				\draw[rotate=135,color=Magenta] (ccv) arc (90:270:4.5 and 2.82);


				\draw (uuu) coordinate[c2,fill=blue,draw=blue,label=above:{$\{{\color{blue}1 },2\}$}];
				\draw (uuv) coordinate[c2,fill=green,draw=green,label={[xshift=-3.7mm]90:$\{3,{\color{green}4}\}$}];
				\draw (uvu) coordinate[c2,fill=red,draw=red,label={[xshift=+4mm]90:$\{{\color{red}5},6\}$}];
				\draw (uvv) coordinate[c2,fill=cyan,draw=cyan,label=above:{$\{7,{\color{cyan}8}\}$}];

				\draw (ucc) coordinate[c2,fill=red,draw=red,label=below right:{$\{{\color{red}5},7\}$}];
				\draw (vcc) coordinate[c2,fill=cyan,draw=cyan,label=above right:{$\{6,{\color{cyan}8}\}$}];
				\draw (ucd) coordinate[c2,fill=blue,draw=blue,label=below right:{$\{{\color{blue}1 },3\}$}];
				\draw (vcd) coordinate[c2,fill=green,draw=green,label={[xshift=+5mm,distance=10mm]90:$\{2,{\color{green}4} \}$}];
				
				\draw (cuc) coordinate[c2,fill=cyan,draw=cyan,label=above right:{$\{3,{\color{cyan}8}\}$}];
				\draw (cvc) coordinate[c2,fill=red,draw=red,label=below right:{$\{2,{\color{red}5}\}$}];
				\draw (cud) coordinate[c2,fill=green,draw=green,label=above right:{$\{{\color{green}4},7\}$}];
				\draw (cvd) coordinate[c2,fill=blue,draw=blue,label={[xshift=+5mm,distance=10mm]270:$\{{\color{blue}1 },6\}$}];
				
				\draw (ccu) coordinate[c2,fill=green,draw=green,label=below:{$\{{\color{green}4},6\}$}];
				\draw (ccv) coordinate[c2,fill=blue,draw=blue,label={[xshift=-3.5mm]270:$\{{\color{blue}1 },7\}$}];
				\draw (cdu) coordinate[c2,fill=cyan,draw=cyan,label={[xshift=+3.5mm]270:$\{2,{\color{cyan}8}\}$}];
				\draw (cdv) coordinate[c2,fill=red,draw=red,label=below:{$\{3,{\color{red}5}\}$}];

			\end{tikzpicture}
		\end{center}
		\caption{\label{2020-f-GHZ}
			Coloring scheme of the ``tight GHZ'' logic~\cite{svozil-2020-ghz} from two-valued states.
			The set-theoretic representation is in terms of
			the canonical partition logic as an equipartitioning of the set $\{1,2,\ldots,8\}$
			obtained from all eight two-valued states on this gadget.
		}
	\end{figure}
	
	\section{A counterexample: Greechie's $G_{32}$}\label{2021-chroma-G32}
	
	It is quite straightforward to demonstrate that the logic $G_{32}$ introduced by Greechie~\cite[Figure~6, p.~121]{greechie:71}
	(see also Refs.~\cite{Holland1975,Bennett-MC-1970,Greechie1974,Greechie-Suppes1976})
	whose hypergraph is depicted in Figure~\ref{2020-f-GreechieG32}(a) has a chromatic number larger than three;
	and, in particular, while having a separating and a unital set of two-valued states, cannot be colored by two-valued states in the algorithmic way proposed earlier.
	Consider the set of all six two-valued states which can be tabulated by the Travis matrix
	\begin{equation}
		T_{ij}=\begin{pmatrix}
			1&0&0&1&0&1&0&0&1&0&0&0&0&0&1\\
			1&0&0&0&1&0&0&1&0&1&0&0&0&1&0\\
			0&1&0&1&0&0&1&0&0&0&1&0&0&1&0\\
			0&1&0&0&1&0&0&0&1&0&0&1&1&0&0\\
			0&0&1&0&0&1&0&1&0&0&1&0&1&0&0\\
			0&0&1&0&0&0&1&0&0&1&0&1&0&0&1
		\end{pmatrix}
		.
	\end{equation}
	There is no way how three of these six row vectors add up to
	a vector whose components are all one; that is,
	$
	\begin{pmatrix}
		1,1,1,1,1,1,1,1,1,1,1,1,1,1,1
	\end{pmatrix}
	$.
	``Completing'' the partition logic and ``extending''
	$G_{32}$ by adding five more contexts
	$\{\{1, 2\}, \{3, 6\}, \{4, 5\}\}$,
	$\{\{1, 4\}, \{2, 3\}, \{5, 6\}\}$,
	$\{\{1, 3\}, \{2, 5\}, \{4, 6\}\}$,
	$\{\{1, 5\}, \{2, 6\}, \{3, 4\}\}$, and
	$\{\{1, 6\}, \{2, 4\}, \{3, 5\}\}$
	does not change the set of two-valued states and thus the Travis matrix.

	Another way of seeing this is to associate a color to, say, the first state.
	As a consequence, all other states, namely states number
	$2$,
	$3$,
	$4$,
	$5$, and
	$6$, need to be eliminated,
	leaving no state which can be associated with
	another color.
	
	One possibility for finding a proper coloring is to drop ``exclusivity'', or rather, the unique association of two-valued states with colors; but not entirely. This can be achieved by not eliminating two-valued states if they appear in association with previous colors. A construction identifying state numbers 1 with red, 3 with green, 5 with blue, and then 2 or four with cyan is depicted in Figure~\ref{2020-f-GreechieG32}(b).
	
	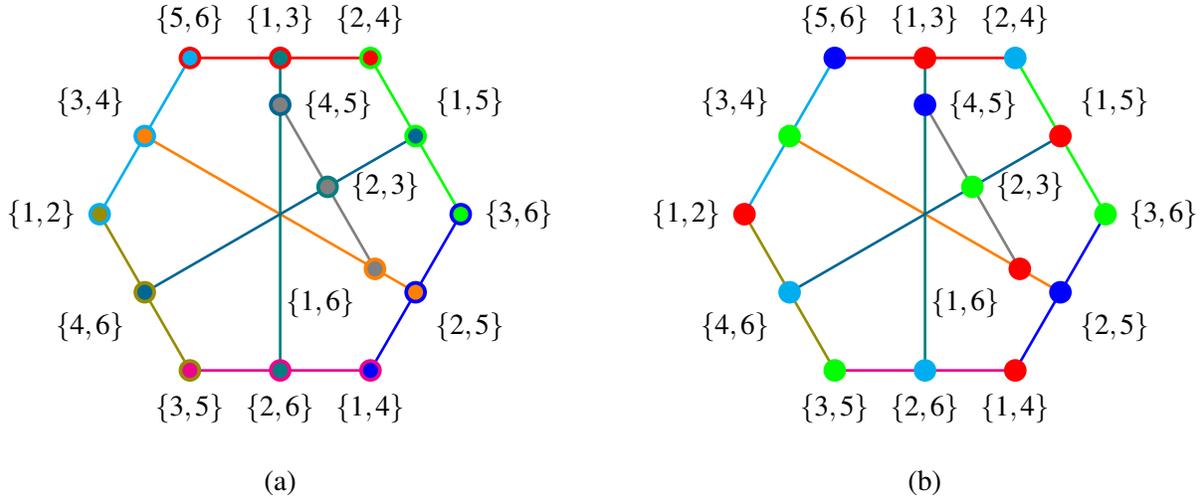
\begin{figure}
		\begin{center}
			\begin{tabular}{ c c c }
				\begin{tikzpicture}  [scale=0.8]
					
					\newdimen\ms
					\ms=0.05cm
					
					\tikzstyle{every path}=[line width=1pt]
					
					\tikzstyle{c3}=[circle,inner sep={\ms/8},minimum size=6*\ms]
					\tikzstyle{c2}=[circle,inner sep={\ms/8},minimum size=4*\ms]
					\tikzstyle{c1}=[circle,inner sep={\ms/8},minimum size=0.8*\ms]
					
					\newdimen\R
					\R=30mm     
					
					
					
					\path
					({ 180 - 0 * 360 /6}:\R      ) coordinate(1)
					({ 180 - 30 - 0 * 360 /6}:{\R * sqrt(3)/2}      ) coordinate(2)
					({ 180 - 1 * 360 /6}:\R   ) coordinate(3)
					({ 180 - 30 - 1 * 360 /6}:{\R * sqrt(3)/2}   ) coordinate(4)
					({ 180 - 2 * 360 /6}:\R  ) coordinate(5)
					({ 180 - 30 - 2 * 360 /6}:{\R * sqrt(3)/2}  ) coordinate(6)
					({ 180 - 3 * 360 /6}:\R  ) coordinate(7)
					({ 180 - 30 - 3 * 360 /6}:{\R * sqrt(3)/2}  ) coordinate(8)
					({ 180 - 4 * 360 /6}:\R     ) coordinate(9)
					({ 180 - 30 - 4 * 360 /6}:{\R * sqrt(3)/2}     ) coordinate(10)
					({ 180 - 5 * 360 /6}:\R     ) coordinate(11)
					({ 180 - 30 - 5 * 360 /6}:{\R * sqrt(3)/2}     ) coordinate(12)
					;
					
					
					\draw [color=cyan] (1) -- (2) -- (3);
					\draw [color=red] (3) -- (4) -- (5);
					\draw [color=green] (5) -- (6) -- (7);
					\draw [color=blue] (7) -- (8) -- (9);
					\draw [color=magenta] (9) -- (10) -- (11);    %
					\draw [color=olive] (11) -- (12) -- (1);    %
					\draw [color=orange] (2) -- (8)  coordinate[pos=0.85]  (15);
					\draw [color=teal] (4) -- (10)  coordinate[pos=0.15]  (13);
					\draw [color=MidnightBlue] (6) -- (12)  coordinate[pos=0.325]  (14);
					\draw [color=gray] (13) --(15);
					
					%
					%
					\draw (1) coordinate[c3,fill=cyan,label={left: $\{ 1,2\} $}];   %
					\draw (1) coordinate[c2,fill=olive];  %
					\draw (2) coordinate[c3,fill=cyan,label={above left: $\{ 3,4\}$}];    %
					\draw (2) coordinate[c2,fill=orange];    %
					\draw (3) coordinate[c3,fill=red,label={above: $\{ 5,6\} $}]; %
					\draw (3) coordinate[c2,fill=cyan];  %
					\draw (4) coordinate[c3,fill=red,label={above: $\{ 1,3\}$}];  %
					\draw (4) coordinate[c2,fill=teal];  %
					\draw (5) coordinate[c3,fill=green,label={above: $\{ 2,4\} $}];  %
					\draw (5) coordinate[c2,fill=red];  %
					\draw (6) coordinate[c3,fill=green,label={above right: $\{ 1,5\} $}];
					\draw (6) coordinate[c2,fill=MidnightBlue];
					\draw (7) coordinate[c3,fill=blue,label={right: $\{ 3,6\}$}];  %
					\draw (7) coordinate[c2,fill=green];  %
					\draw (8) coordinate[c3,fill=blue,label={below right: $\{ 2,5\}$}];  %
					\draw (8) coordinate[c2,fill=orange];  %
					\draw (9) coordinate[c3,fill=magenta,label={below: $\{ 1,4\}$}];
					\draw (9) coordinate[c2,fill=blue];  %
					\draw (10) coordinate[c3,fill=magenta,label={below: $\{ 2,6\}$}];  %
					\draw (10) coordinate[c2,fill=teal];  %
					\draw (11) coordinate[c3,fill=olive,label={below: $\{ 3,5\}$}];  %
					\draw (11) coordinate[c2,fill=magenta];  %
					\draw (12) coordinate[c3,fill=olive,label={below left: $\{ 4,6\}$}];
					\draw (12) coordinate[c2,fill=MidnightBlue];
					\draw (13) coordinate[c3,fill=MidnightBlue,label={right: $\{ 4,5\}$}];  %
					\draw (13) coordinate[c2,fill=gray];  %
					\draw (14) coordinate[c3,fill=teal,label=0:{$\{ 2,3\}$}];  %
					\draw (14) coordinate[c2,fill=gray];  %
					\draw (15) coordinate[c3,fill=orange,label={below left: $\{1,6\}$}];  %
					\draw (15) coordinate[c2,fill=gray];  %
				\end{tikzpicture}
				&$\qquad$&
				\begin{tikzpicture}  [scale=0.8]
					
					\newdimen\ms
					\ms=0.05cm
					
					\tikzstyle{every path}=[line width=1pt]
					
					\tikzstyle{c3}=[circle,inner sep={\ms/8},minimum size=6*\ms]
					\tikzstyle{c2}=[circle,inner sep={\ms/8},minimum size=4*\ms]
					\tikzstyle{c1}=[circle,inner sep={\ms/8},minimum size=0.8*\ms]
					
					\newdimen\R
					\R=30mm     
					
					
					
					\path
					({ 180 - 0 * 360 /6}:\R      ) coordinate(1)
					({ 180 - 30 - 0 * 360 /6}:{\R * sqrt(3)/2}      ) coordinate(2)
					({ 180 - 1 * 360 /6}:\R   ) coordinate(3)
					({ 180 - 30 - 1 * 360 /6}:{\R * sqrt(3)/2}   ) coordinate(4)
					({ 180 - 2 * 360 /6}:\R  ) coordinate(5)
					({ 180 - 30 - 2 * 360 /6}:{\R * sqrt(3)/2}  ) coordinate(6)
					({ 180 - 3 * 360 /6}:\R  ) coordinate(7)
					({ 180 - 30 - 3 * 360 /6}:{\R * sqrt(3)/2}  ) coordinate(8)
					({ 180 - 4 * 360 /6}:\R     ) coordinate(9)
					({ 180 - 30 - 4 * 360 /6}:{\R * sqrt(3)/2}     ) coordinate(10)
					({ 180 - 5 * 360 /6}:\R     ) coordinate(11)
					({ 180 - 30 - 5 * 360 /6}:{\R * sqrt(3)/2}     ) coordinate(12)
					;
					
					
					\draw [color=cyan] (1) -- (2) -- (3);
					\draw [color=red] (3) -- (4) -- (5);
					\draw [color=green] (5) -- (6) -- (7);
					\draw [color=blue] (7) -- (8) -- (9);
					\draw [color=magenta] (9) -- (10) -- (11);    %
					\draw [color=olive] (11) -- (12) -- (1);    %
					\draw [color=orange] (2) -- (8)  coordinate[pos=0.85]  (15);
					\draw [color=teal] (4) -- (10)  coordinate[pos=0.15]  (13);
					\draw [color=MidnightBlue] (6) -- (12)  coordinate[pos=0.325]  (14);
					\draw [color=gray] (13) --(15);
					
					%
					%
					\draw (1) coordinate[c3,fill=red,label={left: $\{ 1,2\} $}];   %
					\draw (2) coordinate[c3,fill=green,label={above left: $\{ 3,4\}$}];    %
					\draw (3) coordinate[c3,fill=blue,label={above: $\{ 5,6\} $}]; %
					\draw (4) coordinate[c3,fill=red,label={above: $\{ 1,3\}$}];  %
					\draw (5) coordinate[c3,fill=cyan,label={above: $\{ 2,4\} $}];  %
					\draw (6) coordinate[c3,fill=red,label={above right: $\{ 1,5\} $}];
					\draw (7) coordinate[c3,fill=green,label={right: $\{ 3,6\}$}];  %
					\draw (8) coordinate[c3,fill=blue,label={below right: $\{ 2,5\}$}];  %
					\draw (9) coordinate[c3,fill=red,label={below: $\{ 1,4\}$}];
					\draw (10) coordinate[c3,fill=cyan,label={below: $\{ 2,6\}$}];  %
					\draw (11) coordinate[c3,fill=green,label={below: $\{ 3,5\}$}];  %
					\draw (12) coordinate[c3,fill=cyan,label={below left: $\{ 4,6\}$}];
					\draw (13) coordinate[c3,fill=blue,label={right: $\{ 4,5\}$}];  %
					\draw (14) coordinate[c3,fill=green,label=0:{$\{ 2,3\}$}];  %
					\draw (15) coordinate[c3,fill=red,label={below left: $\{1,6\}$}];  %
				\end{tikzpicture}\\
				(a)&&(b)
			\end{tabular}
		\end{center}
		\caption{\label{2020-f-GreechieG32}
			(a) Greechie diagram of $G_{32}$ introduced by Greechie~\cite[Figure~6, p.~121]{greechie:71}.
			The overlaid set theoretic representation is in terms of
			the canonical partition logic as an equipartitioning of the set $\{1,2,3,4,5,6\}$
			obtained from all 6 two-valued states on $G_{32}$;
			(b) a nonunique coloring by four colors.
		}
	\end{figure}
	
	Because of the following proof by contradiction, $G_{32}$  cannot have a faithful orthogonal representation:
	Suppose $G_{32}$  has a faithful orthogonal representation.
	Then each one of the nine biconnected contexts of $G_{32}$ can be  uniformly represented by a
	maximal operator~\cite[\S~84, p.~171,172]{halmos-vs}. In order for a faithful orthogonal
	representation to exist the spectral decomposition of two ``intertwining'' maximal
	operators must have (i) (at least) one common projector (ii) with identical eigenvalues which can be identified with identical colors.
	If the logic can be consistently ``covered'' or colored by three colors then the eigenvalues associated with the maximal operators can be the same--that is, these three values (or colors) would occur uniformly in all nine contexts.
	However, this is not the case for $G_{32}$. Therefore, a uniform representation cannot be be given in terms of nine maximal operators with just three eigenvalues per context (that is, maximal operator).
	This is a form of nonclassicality based on a chromatic number exceeding the 2-section's clique number.
	In this case Brooks' theorem~\cite{Brooks1941,Lovasz1975} yields an upper bound of 4 for the chromatic number of $G_{32}$.

	\begin{acknowledgments}
		The authors would like to thank the anonymous referees for their  invaluable comments on the earlier version of this paper.
		
	M.H.S. acknowledges the Institute for Theoretical Physics, Vienna University of Technology for the hospitality during his visit in September 2019, without which writing this paper would not have been possible.
		
		K.S. was supported, in whole, or in part, by the Austrian Science Fund (FWF), Project No. I 4579-N. For the purpose of open access, the author has applied a CC BY public copyright licence to any Author Accepted Manuscript version arising from this submission.

The authors acknowledge TU Wien Bibliothek for financial support through its Open Access Funding Programme.

		The authors kindly acknowledge enlightening discussions with Adan Cabello Jos\'{e}, R. Portillo,
		Alexander~Svozil, Josef Tkadlec and Sebastian Matkovich.
		The authors are grateful to Josef Tkadlec for providing a {\em Pascal} program
		which computes and analyses the set of two-valued states of collections of contexts.
		All misconceptions and errors are of the authors.
	\end{acknowledgments}
	
	\section*{Author Declarations}
	The authors have no conflicts to disclose.
	
	\section*{Data Availability}
	Data sharing is not applicable to this article as no new data were created or analyzed in this study.


%

\end{document}